\documentclass[a4paper,11pt]{amsart}
\usepackage{amssymb}
\usepackage{amscd}
\usepackage{amsmath,amsthm,txfonts}
\usepackage[colorlinks=true]{hyperref}
\usepackage{enumerate}
\usepackage{booktabs,multirow}
\usepackage{tikz}
\usetikzlibrary{patterns}

\allowdisplaybreaks[1]
\title[Canonical bases of a coideal algebra]
{Canonical bases of a coideal subalgebra in $U_q(\mathfrak{sl}_2)$}
\author[K.~Shigechi]{Keiichi~Shigechi}
\email{k1.shigechi at gmail.com}
\date{\today}

\newtheorem{theorem}{Theorem}[section]
\newtheorem{example}[theorem]{Example}
\newtheorem{lemma}[theorem]{Lemma}
\newtheorem{defn}[theorem]{Definition}
\newtheorem{prop}[theorem]{Proposition}
\newtheorem{cor}[theorem]{Corollary}

\newtheorem{conj}[theorem]{Conjecture}
\newtheorem{remark}[theorem]{Remark}

\begin{document}
\begin{abstract}
We consider tensor products of finite-dimensional representations 
of a coideal subalgebra in $U_{q}(\mathfrak{sl}_2)$.
We present an explicit expression for the dual of the canonical bases 
through a diagrammatic presentation.
We show that the decomposition of tensor products of dual canonical bases 
and the action of the coideal subalgebra have integral and positive properties.
As an application, we consider the eigensystem of the generator of the coideal 
subalgebra on the dual canonical bases. 
We provide all the eigenvalues and obtain an explicit expression of the 
eigenfunction for the largest eigenvalue.
The sum of the components of this eigenfunction is conjectured to be 
equal to the total number of arrangements of bishops with a certain symmetry.
\end{abstract}

\maketitle

\section{Introduction}

In~\cite{Lus90-1}, Lusztig introduced the notion of the canonical basis of 
the $q$-analogue of enveloping algebras $U_q(\mathfrak{g})$ 
associated with a simple finite-dimensional Lie algebra $\mathfrak{g}$. 
This basis is characterized by three conditions: the basis is integral,
bar-invariant and spans a $\mathbb{Z}[q^{-1}]$-lattice $\mathcal{L}$ with a 
certain image in the quotient $\mathcal{L}/q^{-1}\mathcal{L}$.
In~\cite{Kas90,Kas91,Kas93}, Kashiwara introduced the notion of (global) 
crystal bases and showed its existence and uniqueness. 
The coincidence of two concepts, the canonical basis and the global 
crystal basis, was shown in ~\cite{Lus90-2,GroLus92}. 
Lusztig constructed a canonical basis in the tensor product and proved its 
associativity in~\cite{Lus92}.
In the case of $\mathfrak{g}=\mathfrak{sl}_2$, Frenkel and Khovanov provided a 
diagrammatic depiction of the dual of the canonical basis in the tensor 
products of finite-dimensional irreducible representations and gave 
the action of the quantum group on these bases~\cite{FK97}.
This diagrammatic method together with the positive integral property
led to the categorification of $U_{q}(\mathfrak{sl}_2)$~\cite{BFK99,FKS06}.

Let $\theta$ be an involution of $\mathfrak{g}$ and $\mathfrak{g}^{\theta}$ be 
the fixed Lie subalgebra. 
The symmetric pair $(\mathfrak{g},\mathfrak{g}^{\theta})$ forms symmetric 
spaces in the classical case. 
In \cite{Nou96,NouDijSug97,NouSug95}, Noumi, Sugitani and Dijkhuizen constructed 
quantum symmetric spaces by using the solution of the reflection equation. 
In ~\cite{Let99,Let02,Let03}, Letzter constructed quantum symmetric spaces by 
using the involution $\theta$ and described the generators of the $q$-analogue
of $U(\mathfrak{g}^{\theta})$ which is a coideal subalgebra of $U_q(\mathfrak{g})$. 
These two approaches produce the same coideal subalgebras~\cite{Let99}.
A general theory for quantum symmetric spaces in the case of symmetrizable 
Kac--Moody algebras was developed in \cite{Kol14}.
In the study of Kazhdan--Lusztig theory of type B, Bao and Wang introduced the notion 
of the quasi-$R$-matrix and canonical bases ($\iota$-canonical bases 
in~\cite{BaoWang13}) for the quantum symmetric pair in the case of 
$\mathfrak{g}=\mathfrak{sl}_n$ \cite{BaoWang13}.	

In this paper, we consider tensor products of finite-dimensional representations 
of a coideal subalgebra $U$ in $U_{q}(\mathfrak{sl}_2)$.
We present an explicit expression for the dual of the canonical bases 
and provide the action of the coideal subalgebra on these bases.
The diagrammatic presentation of the dual bases is also provided. 
For this presentation, we make use of a diagrammatic presentation for Kazhdan--Lusztig 
bases of Hecke algebra of type B studied in~\cite{Shi14}.
We show that the expansion coefficients of a (dual) canonical basis in terms 
of standard bases are written in terms of Kazhdan--Lusztig polynomials. 
We also show that the decomposition of tensor products of dual canonical bases 
and the action of the coideal subalgebra have integral and positive properties.
As an application, we consider the eigensystem of the generator of the coideal 
subalgebra on the dual canonical bases. 
In quantum integrable systems, a coideal subalgebra is the symmetry of the system 
with a boundary, that is, the generators of a coideal subalgebra commute with the 
Hamiltonian.
Therefore, a knowledge of the eigensystem of the generators of $U$ is 
important to study the eigensystem of the Hamiltonian.
We provide all the eigenvalues of the generator of the coideal subalgebra on 
the dual canonical bases. 
An explicit expression of the eigenfunction $\Psi$ for the largest (at $q=1$) 
eigenvalue is obtained (see Theorem~\ref{thm-generic-psi}).
We show that this eigenfunction $\Psi$ has a positive integral property, {\it i.e.}, 
$\Psi\in\mathbb{N}[q,q^{-1}]$.
From this observation, we have a conjecture that the sum of components of $\Psi$
at $q=1$ is equal to the total number of arrangements of bishops with 
a symmetry.

The paper is organized as follows. 
In section \ref{Sec:Hecke}, we briefly recall the parabolic Kazhdan--Lusztig 
polynomials in the case of the Hecke algebra of type B. 
In section \ref{Sec:QG}, we review the definitions and results about the quantum 
group $U_q(\mathfrak{sl}_2)$ and its coideal subalgebra $U$.
In Section \ref{Sec:CB}, we introduce the notion of canonical bases for both 
$U_{q}(\mathfrak{sl}_2)$ and $U$.
The graphical depiction of the dual canonical basis is presented.
In Section \ref{Sec:Comb}, we introduce another graphical method to connect 
canonical bases with Kazhdan--Lusztig polynomials. 
We extend the diagrammatic rules in ~\cite[Section 3]{Shi14} and provide a 
new inversion formula regarding Kazhdan--Lusztig polynomials. 
Section \ref{Sec:IS} is devoted to an analysis of integral and positive 
properties of dual canonical bases. 
We study the eigensystem of the generator of $U$ in details and 
obtain an explicit formula for the eigenvector $\Psi$.
A conjecture on this eigenvector is presented.
In Appendix \ref{Sec:app}, we collect two technical lemmas used in 
Section \ref{Sec:IS}.

\section{\texorpdfstring{Hecke algebra of type $B$ and Kazhdan--Lusztig 
polynomials}{Hecke algebra of type B and Kazhdan--Lusztig polynomials}}
\label{Sec:Hecke}

Let $\mathcal{S}_N$ be the symmetric group and 
$\mathcal{S}_N^C$ be the Weyl group associated with the Dynkin diagram 
of type $C$.
We use a partial order in $\mathcal{S}_{N}^C$, the (strong) {\it Bruhat order}.
We write $w'\le w$ if and only if $w'$ can be obtained as a subexpression
of a reduced expression of $w$.
The {\em Hecke algebra}\/ $\mathcal{H}_N$ of type $B$ is the unital,  
associative algebra over the ring $R:=\mathbb{Z}[t,t^{-1}]$ 
with generators $T_i$, $i=1,\ldots,N$, and relations
\begin{align*}
 (T_i-t)(T_i+t^{-1})&=0 && 1\le i\le N, \\
 T_iT_{i+1}T_i&=T_{i+1}T_iT_{i+1} && 1\le i<N-1,  \\
 T_{N-1}T_NT_{N-1}T_N&=T_NT_{N-1}T_NT_{N-1}, \\
 T_iT_j&=T_jT_i && |i-j|>1.
\end{align*}
The Hecke algebra $\mathcal{H}_N$  has standard basis 
$(T_w)_{w\in \mathcal{S}_{N}^C}$ where 
$T_w=T_{i_1}T_{i_2}\cdots T_{i_r}$ for a reduced word 
$w=s_{i_1}s_{i_2}\cdots s_{i_r}$ written in terms of elementary transpositions $s_i$. 
The involutive ring automorphism of $\mathcal{H}_N$, $a\rightarrow \bar{a}$, 
is defined by $T_i\rightarrow T^{-1}_i$ and $t\rightarrow t^{-1}$. 
Then 
\begin{theorem}[Kazhdan and Lusztig~\cite{KL79}]
There exists a unique basis $C_{w}$ such that $\overline{C_w}=C_{w}$ and 
\begin{eqnarray*}
C_w=\sum_{v\le w}P_{v,w}(t^{-1})T_v,
\end{eqnarray*}
where $P_{v,w}(t^{-1})\in t^{-1}\mathbb{Z}[t^{-1}]$ and $P_{v,v}=1$.
\end{theorem}

Let $\epsilon\in\{+,-\}$. 
We have simple bijections among the following three sets~\cite[Section 2]{Shi14}. 
\begin{enumerate}
\item An element of $\mathcal{S}_N^C/\mathcal{S}_N$.  
\item A binary string in $\{+,-\}^N$. 
\item A path from $(0,0)$ to $(N,n)$ with $|n|\le N$ and 
$N-n\in2\mathbb{Z}$ where each step is in the direction $(1,\pm1)$.
\end{enumerate}
Bijections are realized by the natural action of $\mathcal{S}_N^C$ on 
$\{+,-\}^N$ with representative 
$(+\ldots+)$ for $\epsilon=+$ and $(-\ldots-)$ for $\epsilon=-$.
A sequence $\mathbf{v}=(v_1,\ldots v_N)\in\{+,-\}^N$ is identified with the path with 
the $i$-th step $(1,v_i)$.
We denote by $\mathcal{P}_{N}$ 
the set of paths from $(0,0)$ to $(N,n)$.

\begin{defn}
Let $\alpha,\beta$ be two paths in $\mathcal{P}_{N}$. 
Then, $\alpha\le \beta$ if and only if $\alpha$ is below $\beta$ for $\epsilon=-$, 
above $\beta$ for $\epsilon=+$.
\end{defn}
Note that this definition is compatible with the induced Bruhat order 
in $\mathcal{S}^C_N/\mathcal{S}_{N}$.

We define a free $R$-module $\mathcal{M}_{N}$ with a basis 
indexed by $\mathcal{P}_{N}$, that is,   
$\mathcal{M}_{N}:=\langle m_v: v\in \mathcal{P}_{N}\rangle$.
We have two modules $\mathcal{M}_{N}^\epsilon$, $\epsilon=\pm$, 
corresponding to two natural projection maps from 
$\mathbb{C}[\mathcal{S}_N^C]$ to 
$\mathbb{C}[\mathcal{S}_N^C/\mathcal{S}_N]$~\cite{Deo87}.
The action of $\mathcal{H}_N$ on the modules $\mathcal{M}_{N}^{\epsilon}$
is as follows:
\begin{eqnarray*}
\epsilon=+:\quad 
&&T_im_{\ldots\alpha\alpha\ldots}
= t m_{\ldots\alpha\alpha\ldots}, \quad \alpha=\pm,\quad 1\le i\le N-1 \\
&&T_im_{\ldots+-\ldots}=m_{\ldots-+\ldots}, \quad 1\le i\le N-1\\
&&T_im_{\ldots-+\ldots}=(t-t^{-1})m_{\ldots-+\ldots}+m_{\ldots+-\ldots}, \quad 1\le i\le N-1\\
&&T_Nm_{\ldots+}=m_{\ldots-}, \\
&&T_Nm_{\ldots-}=m_{\ldots+}+(t-t^{-1})m_{\ldots-}, \\ 
\epsilon=-:\quad
&&T_im_{\ldots\alpha\alpha\ldots}
=-t^{-1}m_{\ldots\alpha\alpha\ldots}, \quad \alpha=\pm, \quad 1\le i\le N-1\\
&&T_im_{\ldots-+\ldots}=m_{\ldots+-\ldots}, \quad 1\le i\le N-1\\
&&T_im_{\ldots+-\ldots}
=(t-t^{-1})m_{\ldots+-\ldots}+m_{\ldots-+\ldots}.\quad 1\le i\le N-1 \\
&&T_Nm_{\ldots-}=m_{\ldots+}, \\
&&T_Nm_{\ldots+}=m_{\ldots-}+(t-t^{-1})m_{\ldots+}, \\ 
\end{eqnarray*}
Note that the module $\mathcal{M}^{+}_{N}$ (resp. $\mathcal{M}^{-}_{N}$) 
has a generating vector $m_{+\ldots+}$ (resp. $m_{-\ldots-}$).

We introduce the parabolic analogue of the Kazhdan--Lusztig bases and 
polynomials.
We are interested in the maximal parabolically induced module 
$\mathcal{M}^{\epsilon}_{N}$.
\begin{theorem}[Deodhar~\cite{Deo87}]
There exists a unique basis 
$(C_\alpha^\pm)_{\alpha\in\mathcal{P}_{N}}$
of $\mathcal{M}_{N}^{\pm}$ such that $\overline{C_{\alpha}^\pm}=C_{\alpha}^\pm$ and  
\begin{eqnarray*}
C_\beta^{\pm}=\sum_{\alpha\le\beta}P^{\pm}_{\alpha,\beta}(t^{-1})m_{\alpha}
\end{eqnarray*}
where $\alpha\le\beta$ is in the order of paths associated with the sign 
$\epsilon=\pm$ and the polynomials 
$P^{\pm}_{\alpha,\beta}(t^{-1})\in t^{-1}\mathbb{Z}[t^{-1}]$ 
if $\alpha<\beta$ and $P^{\pm}_{\alpha,\alpha}(t^{-1})=1$.
\end{theorem}

Our definition of $P^{\pm}_{\alpha,\beta}$ differs from the original parabolic 
Kazhdan--Lusztig polynomials by the factor $t^{-d}$ for some $d\in\mathbb{N}$.
The polynomial $P^{-}(t^{-1})$ is a monomial of $t^{-1}$~\cite{Bre09}.  
The algorithm to compute $P^{+}_{\alpha,\beta}(t^{-1})$ was found by 
Boe~\cite{Boe}.
See~\cite{Shi14} for a unified treatment of $P^{\pm}_{\alpha,\beta}$ in 
terms of paths.

\section{\texorpdfstring{Quantum group $U_q(\mathfrak{sl}_2)$}{Quantum group Uq(sl2)}}
\label{Sec:QG}
In this section, we briefly summarize the quantum group $U_q(\mathfrak{sl}_2)$
and a coideal subalgebra in $U_q(\mathfrak{sl}_2)$. 
We follow the notation used in~\cite{BaoWang13,FK97}.

\subsection{\texorpdfstring{Quantum Group $U_q(\mathfrak{sl}_2)$}{Quantum group Uq(sl2)}}
\label{section-QG}
Let $\mathbb{C}(q)$ be the field of rational functions in an indeterminate $q$.
We denote by $\ \bar{}:\mathbb{C}(q)\rightarrow\mathbb{C}(q)$ the $\mathbb{C}$-algebra
involution such that $q^n\mapsto q^{-n}$ for all $n$.
The quantum group $U_q(\mathfrak{sl}_2)$ is an associative algebra 
over $\mathbb{C}(q)$ with generators $K^{\pm1},E,F$ and relations 
\begin{eqnarray*}
&&KK^{-1}=K^{-1}K=1, \\
&&KEK^{-1}=q^2E, \\
&&KFK^{-1}=q^{-2}F, \\
&&\left[E,F\right]=\frac{K-K^{-1}}{q-q^{-1}}.
\end{eqnarray*}
We introduce the quantum integer $[n]:=(q^n-q^{-n})/(q-q^{-1})$,  
the quantum factorial $[n]!:=\prod_{k=1}^n[k]$ 
and the $q$-analogue for the binomial coefficient 
\begin{eqnarray*}
\genfrac{[}{]}{0pt}{}{n}{m}:=\frac{[n]!}{[m]![n-m]!}.
\end{eqnarray*}
We set $E^{(n)}:=E^n/[n]!$ and $F^{(n)}:=F^n/[n]!$.

We define the two involutions. 
One is {\em the Cartan involution}\/ denoted by $\omega$:
\begin{eqnarray*}
&&\omega(E)=F,\quad \omega(F)=E,\quad 
\omega(K^{\pm1})=K^{\pm1}, \quad \omega(q^{\pm1})=q^{\pm1}, \\
&&\omega(xy)=\omega(y)\omega(x),\quad x,y\in U_q(\mathfrak{sl}_2).
\end{eqnarray*}
The other involution is {\em the bar involution}\/  $\psi$ and defined by
\begin{eqnarray*}
&&\psi(E)=E,\quad \psi(F)=F,\quad 
\psi(K^{\pm1})=K^{\mp1},\quad \psi(q^{\pm1})=q^{\mp1}, \\
&&\psi(xy)=\psi(x)\psi(y),\quad x,y\in U_q(\mathfrak{sl}_2).
\end{eqnarray*}

The irreducible $(n+1)$-dimensional representations $V_n$, $n\ge1$, has a basis 
$\{v_m| -n\le m\le n, m\equiv n \pmod{2}\}$.
The action of $U_q(\mathfrak{sl}_2)$ is 
\begin{eqnarray*}
&&K^{\pm1}v_m=q^{\pm m}v_m, \\
&&Ev_m=\left[\frac{n+m}{2}+1\right]v_{m+2}, \\
&&Fv_m=\left[\frac{n-m}{2}+1\right]v_{m-2}.
\end{eqnarray*}
Note that the bases $\{v_m\}$ are canonical bases in the 
sense of \cite{Lus90-1}.
All $U_q(\mathfrak{sl}_2)$-modules in this paper will be finite-dimensional 
representations of type I. 

We define a bilinear symmetric pairing in $V_n$ by 
$\langle xu,v\rangle=\langle u,\omega(x)v\rangle$ and 
$\langle v_n,v_n\rangle=1$ where $u,v\in V_n$ and 
$x\in U_q(\mathfrak{sl}_2)$.
Let $\{v^m| -n\le m\le n, n\equiv m \pmod{2}\}$ be the 
dual bases of $\{v_m\}$ with respect to $\langle,\rangle$.
The action of $U_q(\mathfrak{sl}_2)$ on the dual basis is given 
explicitly in~\cite{FK97}.

For $\kappa=(\kappa_1,\ldots,\kappa_n)$ with $\kappa_i=\pm1, 1\le i\le n$, 
we define
\begin{eqnarray*}
|\kappa|:=\sum_{i=1}^n \kappa_i, \qquad
||\kappa||_-=\sum_{i<j}\theta(\kappa_i<\kappa_j),
\end{eqnarray*}
where $\theta(P)=1$ if $P$ is true and zero otherwise.
We define the projection $\pi_n: V_1^{\otimes n}\mapsto V_n$ by 
\begin{eqnarray}
\label{projection-dual}
\pi_n(v^{\kappa_1}\otimes\ldots\otimes v^{\kappa_n})
=q^{-||\kappa||_{-}}v^{|\kappa|}.
\end{eqnarray}

The quantum group $U_q(\mathfrak{sl}_2)$ has a Hopf algebra structure with 
comultiplication.
We have two different comultiplications $\Delta_{\pm}$:
\begin{eqnarray*}
\Delta_+(K^{\pm1})&=&K^{\pm1}\otimes K^{\pm1}, \\
\Delta_+(E)&=&E\otimes 1 + K\otimes E, \\
\Delta_+(F)&=&F\otimes K^{-1} + 1\otimes F.
\end{eqnarray*}
and 
\begin{eqnarray*}
\Delta_-(K^{\pm1})&=&K^{\pm1}\otimes K^{\pm1}, \\
\Delta_-(E)&=&E\otimes K^{-1} + 1\otimes E, \\
\Delta_-(F)&=&F\otimes 1 + K\otimes F.
\end{eqnarray*}
We define another comultiplications as 
$\overline{\Delta}_{\pm}(x):=
(\psi\otimes\psi)\Delta_{\pm}(\psi(x)), x\in U_q(\mathfrak{sl}_2)$.
We also have counit and antipode, but we do not need them in this paper.
 
Following \cite{Lus92,Lus93}, we define the quasi-$R$-matrix $\Theta_{\pm}$
associated with $\Delta_{\pm}$:
\begin{eqnarray*}
\Theta_{+}
:=
\sum_{k\ge0}(-1)^{k}q^{-k(k-1)/2}\frac{(q-q^{-1})^k}{[k]!}F^k\otimes E^k, 
\end{eqnarray*}
and $\Theta_-:=\sigma\overline{\Theta}_+$ where $\sigma$ is the permutation of 
the tensor factors.
For finite-dimensional representations $M$ and $N$, all but finitely many 
terms of $\Theta_{\pm}$ act as zero on any given vector $m\otimes n\in M\otimes N$.
The quasi-$R$-matrix has the property such that 
$\Theta_{\pm}\overline{\Delta}_\pm(x)=\Delta_\pm(x)\Theta_{\pm}$
and $\Theta_{\pm}\overline{\Theta}_{\pm}=\overline{\Theta}_{\pm}\Theta_{\pm}=1$.
We also define 
$\Theta^{(3)}=(1\otimes\Delta)\Theta\cdot\Theta_{23}$ and in general
\begin{eqnarray*}
\Theta^{(n)}:=(1\otimes\Delta^{n-2})\Theta\cdot\Theta^{(n-1)}_{2,\ldots,n},
\end{eqnarray*}
where $\Theta=\Theta_{\pm}$ and $\Delta=\Delta_{\pm}$.

Let $M$ be a finite-dimensional $U_q(\mathfrak{sl}_2)$-module of type I, 
$B$ be a $\mathbb{C}(q)$-basis of $M$ and the pair $(M,B)$ be a based module 
as in~\cite[Section 27]{Lus93}.
We define an involution $\psi: M\rightarrow M$ by 
$\psi(ab)=\bar{a}b$ for all $a\in\mathbb{C}(q)$ and $b\in B$.
This involution is compatible with the involution $\psi$ on $U_q(\mathfrak{sl}_2)$
in the sense that $\psi(um)=\psi(u)\psi(m)$ for all $u\in U_q(\mathfrak{sl}_2)$
and $m\in M$.
Suppose $M$ and $N$ are finite-dimensional $U_q(\mathfrak{sl}_2)$-modules of 
type I with the involution $\psi$.
Following \cite[Section 27.3]{Lus93}, we define an involution 
$\psi_{\pm}$ on the tensor product $M\otimes N$: 
\begin{eqnarray*}
\psi_{\pm}(m\otimes n):=\Theta_{\pm}(\psi(m)\otimes\psi(n)),\quad m\in M, n\in N.
\end{eqnarray*}
In general, let $M_i$, $1\le i\le r$, be involutive $U_q(\mathfrak{sl}_2)$-modules.
Then the involution 
$\psi_{\pm}:M_{1}\otimes\ldots\otimes M_{r}\rightarrow M_{1}\otimes\ldots\otimes M_{r}$
is recursively given by 
\begin{eqnarray*}
\psi_{\pm}(m_1\otimes\ldots\otimes m_r)
=
\Theta_{\pm}(\psi_{\pm}(m_1\otimes\ldots\otimes m_{p})
\otimes\psi_{\pm}(m_{p+1}\otimes\ldots\otimes m_{r}))
\end{eqnarray*}
for $1\le p\le r-1$ and $m_i\in M_i$.

\subsection{Coideal subalgebra}
We consider the Dynkin diagram of type $A_1$ and the identity 
involution.
By a general theory of quantum symmetric pairs~\cite{Kol14,Let99}, 
we have coideal subalgebras of $U_q(\mathfrak{sl}_2)$.
A coideal subalgebra $U$ is defined as a polynomial algebra in $X$, 
namely, $U:=\mathbb{C}(q)[X]$.
The pair $(U_q(\mathfrak{sl}_2),U)$ is a quantum symmetric pair.
The coideal subalgebra $U$ has an antilinear bar involution 
$\psi^{\iota}$ such that $\psi^{\iota}(X)=X$ and 
$\psi^{\iota}(q)=q^{-1}$.
There exists an injective $\mathbb{C}(q)$-algebra homomorphism
$\iota:U\rightarrow U_q(\mathfrak{sl}_2)$, $X\mapsto E+qFK^{-1}+K^{-1}$. 
In the dual picture, we consider the generator 
$Y:=\psi(\omega(X))=F+q^{-1}KE+K$. 
The comultiplication $\Delta: U\rightarrow U_q(\mathfrak{sl}_2)\times U$ 
is given by 
\begin{eqnarray}
\Delta(X)&=&K^{-1}\otimes X+qFK^{-1}\otimes 1+ E\otimes1, \\
\label{coproductY}
\Delta(Y)&=&K\otimes Y+q^{-1}KE\otimes 1+F\otimes1.
\end{eqnarray}
Note that $U$ is left coideal since 
$\Delta(U)\subset U_q(\mathfrak{sl}_2)\times U$.

A general theory of constructing the quasi-$R$-matrix for a quantum
symmetric pair was developed in~\cite{BaoWang13}.
We collect the facts about the quasi-$R$-matrix for the quantum 
symmetric pair $(U_q(\mathfrak{sl}_2),U)$ in this subsection. 
See \cite{BaoWang13} for a detailed exposition.

The {\it intertwiner} $\Upsilon_{\pm}$ for the quantum symmetric pair 
$(U_q(\mathfrak{sl}_2),U)$  
satisfy 
\begin{eqnarray}
\label{intertwiner}
\iota(\psi^{\iota}(u))\Upsilon_{+}
=
\Upsilon_{+}\psi(\iota(u)), \qquad u\in U.
\end{eqnarray}
The solution of Eqn.(\ref{intertwiner}) is explicitly given in 
\cite[Section 4]{BaoWang13}.
We have $\Upsilon_+=\sum_{n\ge0}\Upsilon_{+,n}$ with 
$\Upsilon_{+,n}=c_n F^{(n)}$.
The coefficients $c_n$ satisfy the recurrence relation
\begin{eqnarray*}
c_n=-q^{-(n-1)}(q-q^{-1})(q[n-1]c_{n-2}+c_{n-1}).
\end{eqnarray*}
with $c_0=0$ and $c_1=1$.
Similarly, we define $\Upsilon_-=\sum_{n\ge0}\Upsilon_{-,n}$ with 
$\Upsilon_{-,n}=\overline{c_n} E^{(n)}$.
We have normalized $\Upsilon_{\pm}$ such that $\Upsilon_{\pm,0}=1$.

The quasi-$R$-matrix $\Theta^{\iota}_{\pm}$ is defined by 
(see~\cite[Section 3]{BaoWang13})
\begin{eqnarray*}
\Theta^{\iota}_{\pm}
:=
\Delta_{\pm}(\Upsilon_{\pm})\Theta_{\pm}(1\otimes\Upsilon_{\pm}^{-1}).
\end{eqnarray*}
The quasi-$R$-matrix satisfies 
$\Theta^{\iota}_{\pm}\overline{\Theta^{\iota}}_{\pm}=1$ 
and 
$\Delta_{\pm}(u)\Theta^{\iota}_{\pm}
=\Theta^{\iota}_{\pm}\overline{\Delta}_{\pm}(u)$ for 
$u\in U$.

Let $(M,B)$ be a based module of $U_q(\mathfrak{sl}_2)$ as in 
Section~\ref{section-QG}.
We regard $M$ as a $U$-module. 
We define an involution  $\psi^{\iota}_{\pm}: M\rightarrow M$
by $\psi^{\iota}_{\pm}:=\Upsilon_{\pm}\circ\psi$.
The involution $\psi^{\iota}_{\pm}$ is compatible with 
$\psi^{\iota}$ on $U$ in the sense that 
$\psi^{\iota}_{\pm}(um)=\psi^{\iota}(u)\psi^{\iota}_{\pm}(m)$ for 
all $u\in U$ and $m\in M$.
Suppose $M$ is a $U_q(\mathfrak{sl}_2)$-module equipped with $\psi$ 
and $N$ is a $U$-module equipped with $\psi^{\iota}_{\pm}$.
We regard $M\otimes N$ as a $U$-module.
Following \cite[Section 3.4]{BaoWang13}, we define the involution 
$\psi^\iota_{\pm}$ on the tensor product $M\otimes N$:
\begin{eqnarray*}
\psi^\iota_{\pm}(m\otimes n)
:=\Theta^{\iota}_{\pm}(\psi(m)\otimes\psi^\iota_{\pm}(n)), 
\qquad m\in M, n\in N.
\end{eqnarray*}
In general, let $M_i, 1\le i\le r$, be $U_q(\mathfrak{sl}_2)$-modules
and $N$ be a $U$-module.
The involution $\psi_{\iota}^{\pm}$ is recursively given by 
\begin{eqnarray*}
\psi^{\iota}_{\pm}(m_1\otimes\ldots\otimes m_r\otimes n)
=
\Theta^{\iota}_{\pm}(\psi_{\pm}(m_1\otimes\ldots\otimes m_p)
\otimes\psi_{\pm}^{\iota}(m_{p+1}\otimes\ldots\otimes m_r\otimes n)),
\end{eqnarray*}
where $m_i\in M_i$, $1\le p\le r$, and $n\in N$.

\section{\texorpdfstring{Canonical bases of $U_q(\mathfrak{sl}_2)$ and $U$}
{Canonical basis of Uq(sl2) and U}}
\label{Sec:CB}

Let $\mathbf{k}=(k_1,\ldots,k_n)$ and $\mathbf{l}=(l_1,\ldots,l_n)$.
When $\sum_{i=1}^{m}k_i\le\sum_{i=1}^{m}l_i$ for all $1\le m\le n$,
we denote it by $\mathbf{k}\le_{+}\mathbf{l}$ or 
$\mathbf{l}\le_{-}\mathbf{k}$. 
For $\mathbf{m}:=(m_1,\ldots,m_n)\in\mathbb{N}_{+}^{n}$, we define 
\begin{eqnarray*}
I_{\mathbf{m}}:=
\left\{k_i, 1\le i\le n| -m_i\le k_i\le m_i,  k_i\equiv m_{i} (\mathrm{mod}\ 2)\right\}.
\end{eqnarray*}

In this section, we follow the notation and the convention used in \cite{BaoWang14,Lus93}.

\subsection{Canonical bases}
Let $(M,B)$ and $(M',B')$ be based modules. 
The tensor product of two based modules $M\otimes M'$ has a basis 
$B\otimes B'$.
This basis is not compatible with the involution $\psi$ in general.
We introduce a modified basis $B\diamondsuit B'$ in the tensor product 
following~\cite{Lus92}.
We call the basis $B\diamondsuit B'$ a {\it canonical basis}.
More in general, we obtain canonical bases 
$\{v_{k_1}\diamondsuit\cdots\diamondsuit v_{k_n}\}_{\mathbf{k}\in I_{\mathbf{m}}}$ 
in the tensor product $V_{m_1}\otimes\cdots\otimes V_{m_n}$.
Note that we have associativity of tensor products.
The canonical basis is characterized as follows.
\begin{theorem}[Lusztig~\cite{Lus92}]
\leavevmode
\label{theorem-canonical1}
\begin{enumerate}
\item 
There exists a unique element 
$v_{k_1}\diamondsuit\cdots\diamondsuit v_{k_n}\in V_{m_1}\otimes\cdots\otimes V_{m_n}$
such that
\begin{eqnarray*}
&&\psi_{+}(v_{k_1}\diamondsuit\cdots\diamondsuit v_{k_n})
=v_{k_1}\diamondsuit\cdots\diamondsuit v_{k_n}, \\
&&v_{k_1}\diamondsuit\cdots\diamondsuit v_{k_n}
-v_{k_1}\otimes\cdots\otimes v_{k_n}
\in q^{-1}\cdot_{\mathbb{Z}[q^{-1}]}V_{m_1}\otimes\cdots\otimes V_{m_n}.
\end{eqnarray*}
\item
The vector 
$v_{k_1}\diamondsuit\ldots\diamondsuit v_{k_n}-v_{k_1}\otimes\ldots\otimes v_{k_n}$
is a linear combination of $v_{l_1}\otimes\ldots\otimes v_{l_n}$ with 
$\mathbf{l}\neq\mathbf{k}$, $\sum_{i=1}^{n}k_i=\sum_{i=1}^{n}l_i$ and  
$\mathbf{l}\le_{+}\mathbf{k}$. 
The coefficients are in $q^{-1}\mathbb{Z}[q^{-1}]$.
\item 
The elements 
$v_{k_1}\diamondsuit\cdots\diamondsuit v_{k_n}$ 
form a $\mathbb{C}(q)$-basis of $V_{m_1}\otimes\cdots\otimes V_{m_n}$, 
a $\mathbb{Z}[q,q^{-1}]$-basis of 
$_{\mathbb{Z}[q,q^{-1}]}V_{m_1}\otimes\cdots\otimes V_{m_n}$, and 
a $\mathbb{Z}[q^{-1}]$-basis of 
$_{\mathbb{Z}[q^{-1}]}V_{m_1}\otimes\cdots\otimes V_{m_n}$.  
\end{enumerate}
\end{theorem}
We define the bilinear pairing of 
$V_{m_1}\otimes\ldots\otimes V_{m_n}$ and 
$V_{m_1}\otimes\ldots\otimes V_{m_n}$ by
\begin{eqnarray*}
\langle v_{k_1}\otimes\ldots\otimes v_{k_n},
v^{k'_1}\otimes\ldots\otimes v^{k'_n}\rangle
=\delta_{k_1}^{k'_1}\ldots\delta_{k_n}^{k'_n}.
\end{eqnarray*}
We define the dual canonical basis $v^{k_1}\heartsuit\cdots\heartsuit v^{k_n}$
with respect to the bilinear pairing:
\begin{eqnarray*}
\langle v_{k_1}\diamondsuit\cdots\diamondsuit v_{k_n},
v^{l_1}\heartsuit\cdots\heartsuit v^{l_n} \rangle
=\delta_{k_1}^{l_1}\cdots\delta_{k_n}^{l_n}.
\end{eqnarray*}
The dual statement of Theorem ~\ref{theorem-canonical1} is 
\begin{theorem}[Frenkel and Khovanov~{\cite[Theorem 1.8]{FK97}}]
\leavevmode
\begin{enumerate}
\item There exists a unique element  
$v^{k_1}\heartsuit\cdots\heartsuit v^{k_n}\in V_{m_1}\otimes\cdots\otimes V_{m_n}$
such that 
\begin{eqnarray*}
&&\psi_{-}(v^{k_1}\heartsuit\cdots\heartsuit v^{k_n})
=v^{k_1}\heartsuit\cdots\heartsuit v^{k_n}, \\
&& 
v^{k_1}\heartsuit\cdots\heartsuit v^{k_n}
-v^{k_1}\otimes\cdots\otimes v^{k_n}
\in q^{-1}\cdot_{\mathbb{Z}[q^{-1}]}
V_{m_1}\otimes\cdots\otimes V_{m_n}
\end{eqnarray*}
\item 
The vector 
$v^{k_1}\heartsuit\cdots\heartsuit v^{k_n}-v^{k_1}\otimes\cdots\otimes v^{k_n}$
is a linear combination of $v^{l_1}\otimes\cdots\otimes v^{l_n}$ with 
$\mathbf{l}\neq\mathbf{k}$, $\sum_{i=1}^{n}k_i=\sum_{i=1}^{n}l_i$ and 
$\mathbf{k}<_{+}\mathbf{l}$. 
The coefficients are in $q^{-1}\mathbb{Z}[q^{-1}]$.

\item
The elements 
$v_{k_1}\heartsuit\cdots\heartsuit v_{k_n}$ 
form a $\mathbb{C}(q)$-basis of $V_{m_1}\otimes\cdots\otimes V_{m_n}$, 
a $\mathbb{Z}[q,q^{-1}]$-basis of 
$_{\mathbb{Z}[q,q^{-1}]}V_{m_1}\otimes\cdots\otimes V_{m_n}$, and 
a $\mathbb{Z}[q^{-1}]$-basis of 
$_{\mathbb{Z}[q^{-1}]}V_{m_1}\otimes\cdots\otimes V_{m_n}$.  
\end{enumerate}
\end{theorem}

We regard a based module $(M,B)$ as a finite-dimensional $U$-module. 
When $b\in B$ is a lowest weight vector, $\psi^{\iota}_{+}(b)=b$.
Similarly, when $b\in B$ is a highest weight vector, $\psi^{\iota}_{-}(b)=b$.
For any other $b\in B$, we have $\psi^{\iota}_{\pm}(b)\neq b$, that is, 
the basis $B$ is not compatible with $\psi^{\iota}_{\pm}$.
We can introduce a modified basis $B'$ with the property $\psi^{\iota}(b')=b'$ 
for $b'\in B'$.
By abuse of notation, we call the pair $(M,B')$ a based module of $U$ and 
the basis $B'$ a canonical basis of $U$.

Let $(M,B)$ and $(M',B')$ be based modules of $U_q(\mathfrak{sl}_2)$ and 
$U$ respectively.
By a similar argument to canonical bases of $U_q(\mathfrak{sl}_2)$, 
we can introduce a modified basis $B\varclubsuit B'$ in the
tensor product of two based modules $M\otimes M'$.
In general, we obtain bases 
$\{v^{k_1}\varclubsuit\cdots\varclubsuit v^{k_n}\}_{\mathbf{k}\in I_{\mathbf{m}}}$ 
in the tensor product $V_{m_1}\otimes\cdots\otimes V_{m_n}$.
We call $v^{k_1}\varclubsuit\cdots\varclubsuit v^{k_n}$, 
$\mathbf{k}\in I_{\mathbf{m}}$, a canonical basis of $U$.
The canonical bases of $U$ are characterized as follows.
The proofs are similar to the one for Theorem~\ref{theorem-canonical1}.
\begin{theorem}
\leavevmode
\label{theorem-canonical3}
\begin{enumerate}
\item
There exists a unique element 
$v_{k_1}\varclubsuit\ldots\varclubsuit v_{k_n}\in V_{m_1}\otimes\ldots\otimes V_{m_n}$
such that 
\begin{eqnarray*}
&&\psi^{\iota}_{+}(v_{k_1}\varclubsuit\ldots\varclubsuit v_{k_n})
=v_{k_1}\varclubsuit\ldots\varclubsuit v_{k_n}, \\
&&v_{k_1}\varclubsuit\ldots\varclubsuit v_{k_n}
-v_{k_1}\otimes\ldots\otimes v_{k_n}
\in q^{-1}\cdot_{\mathbb{Z}[q^{-1}]}V_{m_1}\otimes\ldots\otimes V_{m_n}.
\end{eqnarray*}
\item 
The vector 
$v_{k_1}\varclubsuit\ldots\varclubsuit v_{k_n}-v_{k_1}\otimes\ldots\otimes v_{k_n}$
is a linear combination of $v_{l_1}\otimes\ldots\otimes v_{l_n}$ with 
$\mathbf{l}\le_{+}\mathbf{k}$. 
The coefficients are in $q^{-1}\mathbb{Z}[q^{-1}]$.

\item 
The elements 
$v_{k_1}\varclubsuit\cdots\varclubsuit v_{k_n}$ 
form a $\mathbb{C}(q)$-basis of $V_{m_1}\otimes\cdots\otimes V_{m_n}$, 
a $\mathbb{Z}[q,q^{-1}]$-basis of 
$_{\mathbb{Z}[q,q^{-1}]}V_{m_1}\otimes\cdots\otimes V_{m_n}$, and 
a $\mathbb{Z}[q^{-1}]$-basis of 
$_{\mathbb{Z}[q^{-1}]}V_{m_1}\otimes\cdots\otimes V_{m_n}$.  
\end{enumerate}
\end{theorem}
We define the dual of a canonical basis,  
$v^{l'_1}\varspadesuit \ldots \varspadesuit v^{l'_n}$, 
with respect to the bilinear pairing:
\begin{eqnarray}
\label{IProductSB}
\langle v_{l_1}\varclubsuit\ldots\varclubsuit v_{l_n} 
v^{l'_1}\varspadesuit \ldots \varspadesuit v^{l'_n} \rangle
=\delta_{l_1}^{l'_1}\ldots\delta_{l_n}^{l'_n}.
\end{eqnarray}
The dual statement of Theorem~\ref{theorem-canonical3} is 
\begin{theorem}
\leavevmode
\begin{enumerate}
\item
There exists a unique element 
$v^{k_1}\varspadesuit\ldots\varspadesuit v^{k_n}\in V_{m_1}\otimes\ldots\otimes V_{m_n}$
such that 	
\begin{eqnarray*}
&&\psi^{\iota}_{-}(v^{k_1}\varspadesuit\ldots\varspadesuit v^{k_n})
=v^{k_1}\varspadesuit\ldots\varspadesuit v^{k_n}, \\
&&v^{k_1}\varspadesuit\ldots\varspadesuit v^{k_n}
-v^{k_1}\otimes\ldots\otimes v^{k_n}
\in q^{-1}\cdot_{\mathbb{Z}[q^{-1}]}V_{m_1}\otimes\ldots\otimes V_{m_n}
\end{eqnarray*}
\item 
The vector 
$v^{k_1}\varspadesuit\ldots\varspadesuit v^{k_n}-v^{k_1}\otimes\ldots\otimes v^{k_n}$
is a linear combination of $v^{l_1}\otimes\ldots\otimes v^{l_n}$ with 
$\mathbf{l}\le_{-}\mathbf{k}$.
The coefficients are in $q^{-1}\mathbb{Z}[q^{-1}]$.

\item 
The elements 
$v_{k_1}\varspadesuit\cdots\varspadesuit v_{k_n}$ 
form a $\mathbb{C}(q)$-basis of $V_{m_1}\otimes\cdots\otimes V_{m_n}$, 
a $\mathbb{Z}[q,q^{-1}]$-basis of 
$_{\mathbb{Z}[q,q^{-1}]}V_{m_1}\otimes\cdots\otimes V_{m_n}$, and 
a $\mathbb{Z}[q^{-1}]$-basis of 
$_{\mathbb{Z}[q^{-1}]}V_{m_1}\otimes\cdots\otimes V_{m_n}$.  
\end{enumerate}
\end{theorem}

\subsection{Graphical depiction of the standard bases}
The graphical calculus for $U_q(\mathfrak{sl}_2)$
was developed in~\cite[Section 2]{FK97}.
Since the comultiplication in the dual space is different from \cite{FK97},
we briefly summarize the graphical calculus for standard bases 
in this subsection.

The dual basis $v^{n-2k}$ is written as 
\begin{eqnarray*}
v^{n-2k}=\pi_n((v^1)^{\otimes(n-k)}\otimes(v^{-1})^{\otimes k}).
\end{eqnarray*}
The diagram for $v^{n-2k}$ is depicted as
\begin{eqnarray*}
\begin{tikzpicture}
\path[draw](0,0)--(2.97,0)--(2.97,-.5)--(0,-.5)--(0,0);
\draw[-](0.3,-1.2)--(0.3,-.5)(0.23,-0.7)--(0.3,-.5)--(0.37,-0.7);
\draw[-](0.3+0.9,-1.2)--(0.3+0.9,-.5)(0.23+0.9,-0.7)--(0.3+0.9,-.5)--(0.37+0.9,-0.7);
\draw[-](1.7,-1.2)--(1.7,-.5)(1.63,-1)--(1.7,-1.2)--(1.77,-1);
\draw[-](2.6,-1.2)--(2.6,-0.5)(1.63+0.9,-1)--(1.7+0.9,-1.2)--(1.77+0.9,-1);
\draw (2.97/2,-0.25) node{$n$} (0.75,-0.85)node{\ldots} (2.15,-0.85)node{\ldots};
\draw (0.75,-1.5)node{$\underbrace{\qquad\quad}_{n-k}$};
\draw (2.15,-1.5)node{$\underbrace{\qquad\quad}_{k}$};
\end{tikzpicture},
\end{eqnarray*}
where the box marked by $n$ with $n$ lines corresponds to 
the projector $\pi_{n}$.
The graph for a tensor product $v^{n_1-2k_1}\otimes\ldots v^{n_r-2k_r}$ 
is depicted by placing the diagram for $v^{n_i-2k_i}, 1\le i\le r$, from left 
to right in parallel. 

\subsection{Graphical depiction of dual canonical bases}
The diagram for $v^{n_1-2k_1}\varspadesuit\ldots\varspadesuit v^{n_r-2k_r}$ is obtained 
from the diagram for $v^{n_1-2k_1}\otimes\ldots\otimes v^{n_r-2k_r}$ 
by the following rules. 
The rules are essentially the same as the ones for 
the Kazhdan--Lusztig basis studied in~\cite{Shi14} ($m=1$ of Case B in~\cite{Shi14}).
\begin{enumerate}[(A)]
\item Make a pair between adjacent down arrow and up arrow (in this order). 
Connect this pair of two arrows into a simple unoriented arc. 
The arc does not intersect with anything.
\item Repeat the above procedure (making pairs) until all the up arrows are 
to the left of all down arrows.
\item Put a star ($\bigstar$) on the rightmost down arrow if it exists.
\item For the remaining down arrows, we make a pair of two adjacent down arrows
from right. 
Connect this pair of two arrows into a simple unoriented dashed arc.
\end{enumerate}
After applying the rules (A)-(D), there may be a down arrow which does not form 
a dashed arc. 
We call this down arrow an {\it unpaired down arrow}.
The diagram for $v^{m_1-2k_1}\heartsuit\ldots\heartsuit v^{m_n-2k_n}$ is obtained
by the rules (A) and (B) ~\cite[Section 2.3]{FK97}.

Each building block (a simple oriented arc, a dashed arc and an arrow with a star)
is a vector in $V_1$ or $V_1\otimes V_1$:
\begin{eqnarray}
\label{building-1}
\raisebox{-0.5\totalheight}{
\begin{tikzpicture}
\draw[thick](0,0)..controls (0,-.8)and(1.0,-.8)..(1.0,0);
\end{tikzpicture}}
&=&v^{-1}\otimes v^{1}-q^{-1}v^{1}\otimes v^{-1}, \\
\label{building-2}
\raisebox{-0.5\totalheight}{
\begin{tikzpicture}
\draw[thick,dashed](0,0)..controls (0,-.8)and(1.0,-.8)..(1.0,0);
\end{tikzpicture}}
&=&v^{-1}\otimes v^{-1}-q^{-1}v^{1}\otimes v^{1}, \\
\label{building-3}
\raisebox{-0.5\totalheight}{
\begin{tikzpicture}
\draw[thick](0,0)--(0,-0.7);
\draw (0,-.8)node{$\bigstar$};
\end{tikzpicture}}
&=&v^{-1}-q^{-1}v^{1},
\end{eqnarray}
and an up arrow (resp. an unpaired down arrow) corresponds to 
$v^{1}$ (resp. $v^{-1}$). 
A vector corresponding to 
$v^{k_1}\varspadesuit\ldots\varspadesuit v^{2k_r}$ 
is obtained by acting the projection on a tensor product of vectors 
for building blocks. 

\begin{example}
The graph and the vector corresponding to 
$v^{-1}\varspadesuit v^{-1}\in V_3\otimes V_3$ 
are as follows:
\begin{eqnarray*}
\raisebox{-0.5\totalheight}{
\begin{tikzpicture}
\draw (0,0)--(1.6,0)--(1.6,-.5)--(0,-0.5)--(0,0);
\draw (0.3,-0.5)--(0.3,-1.2)(0.23,-0.7)--(0.3,-.5)--(0.37,-0.7);
\draw (2,0)--(1.6+2,0)--(1.6+2,-.5)--(2,-0.5)--(2,0);
\draw (2.0+1.3,-0.5)--(2.0+1.3,-1.2);
\draw (1.3,-0.5)..controls(1.3,-1)and(2.3,-1)..(2.3,-0.5);
\draw[dashed] (0.8,-0.5)..controls(0.8,-1.5)and(2.8,-1.5)..(2.8,-0.5);
\draw (3.3,-1.2)node{$\bigstar$};
\end{tikzpicture}
}
&=&[-1,-1]-q^{-2}[1,-3]-q^{-1}[1,1]-q^{-2}[-1,1]+q^{-3}[3,-1]\\
&&+q^{-5}[1,-1]+q^{-2}[1,3]-q^{-5}[3,1].
\end{eqnarray*}	
where $[i,j]:=v^i\otimes v^j$.
\end{example}

First, we describe the dual canonical basis of $V_1^{\otimes n}$.
Set $t=-q$ in the notation of Section 2 and let $\kappa_i=\pm1$ for 
$1\le i\le n$. 
\begin{lemma}
\label{lemma-KL-dC}
The diagram for $v^{\kappa_1}\varspadesuit\ldots\varspadesuit v^{\kappa_n}$
provides the dual canonical basis of $V_1^{\otimes n}$.
\end{lemma}

\begin{proof}
The diagram for 
$v^{\kappa_1}\varspadesuit\ldots\varspadesuit v^{\kappa_n}$
is the same as the diagram for the Kazhdan--Lusztig 
basis in~\cite{Shi14}.
From \cite[Theorem 5.8]{BaoWang13}, there exists an involution
which compatible with both the bar involution on $\mathcal{H}_N$ and 
the bar involution $\psi_{-}^{\iota}$.
Thus, the Kazhdan--Lusztig basis on $\mathcal{M}_N^{-}$ coincides 
with the canonical basis of $V_1^{\otimes N}$.
\end{proof}

For $\mathbf{k}\in I_{\mathbf{m}}$, 
let $\kappa:=(\kappa_1,\ldots,\kappa_{N})$, $N=\sum_{i}^{n}m_i$, be a sequence of 
$+1$ and $-1$ such that it starts with $(m_1+k_1)/2$ copies of $+1$, followed 
by $(m_1-k_1)/2$ copies of $-1$, followed by $(m_2+k_2)/2$ copies of $+1$, 
followed by $(m_2-k_2)/2$ copies of $-1$, $\ldots$ and ends with $(m_n-k_n)/2$ 
copies of $-1$.
We have the following theorem about the dual canonical basis:
\begin{theorem}
\label{Thm-dC}
The diagram for 
$v^{k_1}\varspadesuit\ldots\varspadesuit v^{k_n}$
provides the dual canonical basis of $U$ in 
$V_{m_1}\otimes\ldots\otimes V_{m_n}$, {\it i.e.},  
\begin{eqnarray}
v^{k_1}\varspadesuit\ldots\varspadesuit v^{k_n}
=(\pi_{m_1}\otimes\ldots\otimes\pi_{m_n})
v^{\kappa_1}\varspadesuit\ldots\varspadesuit v^{\kappa_N}.
\end{eqnarray}
\end{theorem}
\begin{proof}
We denote $\pi:=\pi_{m_1}\otimes\ldots\otimes\pi_{m_n}$. 
We will first show that $\pi(v^{\kappa_1}\varspadesuit\ldots\varspadesuit v^{\kappa_N})$ 
is invariant under the action of $\psi^{\iota}_-$. 
The involution $\psi^{\iota}_-$ is written as 
\begin{eqnarray*}
\psi^{\iota}_-=\Delta^{n-1}(\Upsilon_-)\Theta^{(n)}_-\psi\otimes\ldots\otimes\psi.
\end{eqnarray*}
We expand $v^{k_1}\varspadesuit\ldots\varspadesuit v^{k_n}$ in terms of 
the dual canonical basis $v^{k'_1}\heartsuit\ldots\heartsuit v^{k'_n}$:
\begin{eqnarray*}
v^{k_1}\varspadesuit\ldots\varspadesuit v^{k_n}
=
\sum_{\mathbf{k'}}c_{\mathbf{k},\mathbf{k'}}
v^{k'_1}\heartsuit\ldots\heartsuit v^{k'_n},
\end{eqnarray*}
where $c_{\mathbf{k},\mathbf{k'}}\in q^{-1}\mathbb{Z}[q^{-1}]$ for 
$\mathbf{k}\neq\mathbf{k'}$ and $c_{\mathbf{k},\mathbf{k}}=1$.

From Lemma~\ref{lemma-KL-dC}, 
$v^{\kappa_1}\varspadesuit\ldots\varspadesuit v^{\kappa_N}$ is invariant 
under the action of $\psi^{\iota}_-$. 
Together with the fact that 
$v^{\kappa_1}\heartsuit\ldots\heartsuit v^{\kappa_N}$ is invariant 
under the action of $\Theta^{(N)}_{-}\psi\otimes\ldots\otimes\psi$, 
we have 
\begin{eqnarray*}
\Delta^{n-1}(\Upsilon_-)\sum_{\kappa'}
\overline{c_{\kappa,\kappa'}}
v^{\kappa'_1}\heartsuit\ldots\heartsuit v^{\kappa'_N}
=
\sum_{\kappa'}c_{\kappa,\kappa'}
v^{\kappa'_1}\heartsuit\ldots\heartsuit v^{\kappa'_N}.
\end{eqnarray*}
We act with $\pi$ on both sides. 
Since $\pi$ is a projection, the action of $\pi$ commutes 
with the one of $\Delta^{n-1}(\Upsilon_-)$.
Further, the action of $\pi$ on 
$v^{\kappa'_1}\heartsuit\ldots\heartsuit v^{\kappa'_N}$
may vanish by using 
\begin{eqnarray*}
\raisebox{-0.5\totalheight}{
\begin{tikzpicture}
\draw(0,0)--(2,0)--(2,-.5)--(0,-.5)--(0,0);
\draw(1,-.25)node{$n$};
\draw(0.5,-0.5)..controls(0.5,-1.0)and(1.0,-1.0)..(1.0,-0.5);
\end{tikzpicture}}
=0.
\end{eqnarray*}
Note that this vanishing condition is compatible with 
\begin{eqnarray*}
\raisebox{-0.5\totalheight}{
\begin{tikzpicture}
\draw(0,0)--(2,0)--(2,-.5)--(0,-.5)--(0,0);
\draw(1,-.25)node{$n$};
\draw(0.7,-0.5)--(0.7,-1.2)(0.63,-1)--(0.7,-1.2)--(0.77,-1);
\draw(1.3,-0.5)--(1.3,-1.2)(1.23,-0.7)--(1.3,-0.5)--(1.37,-0.7);
\draw(0.2,-0.9)node{$\ldots$}(1.7,-0.9)node{$\ldots$};
\end{tikzpicture}}
=q^{-1}
\raisebox{-0.5\totalheight}{
\begin{tikzpicture}
\draw(0,0)--(2,0)--(2,-.5)--(0,-.5)--(0,0);
\draw(1,-.25)node{$n$};
\draw(0.7,-0.5)--(0.7,-1.2)(0.63,-0.7)--(0.7,-0.5)--(0.77,-0.7);
\draw(1.3,-0.5)--(1.3,-1.2)(1.23,-1)--(1.3,-1.2)--(1.37,-1);
\draw(0.2,-0.9)node{$\ldots$}(1.7,-0.9)node{$\ldots$};
\end{tikzpicture}
}.
\end{eqnarray*}
Therefore, $\pi(v^{\kappa'_1}\heartsuit\ldots\heartsuit v^{\kappa'_N})$
is equal to $v^{k'_1}\heartsuit\ldots\heartsuit v^{k'_n}$ or zero.
For non-zero $\pi(v^{\kappa'_1}\heartsuit\ldots\heartsuit v^{\kappa'_N})$,
we can associate an sequence of integers $\mathbf{k'}$ with $\kappa'$ 
by the inverse of the map from $\mathbf{k}$ to $\kappa$ described just above 
Theorem~\ref{Thm-dC}.
For such $\mathbf{k'}$ and $\kappa'$, we abbreviate $c_{\kappa,\kappa'}$
as $c'_{\mathbf{k},\mathbf{k'}}$.
Note that the coefficients $c'_{\mathbf{k},\mathbf{k'}}$ is nothing but the 
expansion coefficients $c_{\mathbf{k},\mathbf{k'}}$.
Thus, We have 
\begin{eqnarray*}
\Delta^{n-1}(\Upsilon_-)\sum_{\mathbf{k'}}
\overline{c_{\mathbf{k},\mathbf{k'}}}
v^{k'_1}\heartsuit\ldots\heartsuit v^{k'_n}
=
\sum_{\mathbf{k'}}c_{\mathbf{k},\mathbf{k'}}
v^{k'_1}\heartsuit\ldots\heartsuit v^{k'_n}.
\end{eqnarray*}
This shows that $v^{k_1}\varspadesuit\ldots\varspadesuit v^{k_n}$
is invariant under the action of $\psi^{\iota}_-$.

The expansion coefficients of 
$v^{k_1}\varspadesuit\ldots\varspadesuit v^{k_n}$ 
in terms of the standard bases 
$v^{k'_1}\otimes\ldots\otimes v^{k'_n}$
are in $q^{-1}\mathbb{Z}[q^{-1}]$ for $\mathbf{k}\neq\mathbf{k'}$
and one for $\mathbf{k}=\mathbf{k'}$ by construction.
From Eqns.(\ref{building-1}) to (\ref{building-3}), 
a building block of a diagram has the form 
$v^{\kappa_1}\otimes v^{\kappa_2}
+q^{-1}\mathbb{Z}[q^{-1}]v^{\kappa'_1}\otimes v^{\kappa'_2}$ 
with $(\kappa'_1,\kappa'_2)<_{-}(\kappa_1,\kappa_2)$
or $v^{\kappa_1}+q^{-1}\mathbb{Z}[q^{-1}]v^{\kappa'_1}$
with $\kappa'_1<_{-}\kappa_1$.
Thus, if $\mathbf{k}\neq\mathbf{k'}$, we have $\mathbf{k'}<_{-}\mathbf{k}$.
This completes the proof of the theorem.
\end{proof}

\subsection{Inversion formula}
Let $\mathbf{k}=(k_1,\ldots,k_n)$ and 
$\mathbf{l}=(l_1\ldots,l_n)$.
We define 
\begin{eqnarray*}
R_{\mathbf{k},\mathbf{l}}(q^{-1})
&:=&\langle v_{l_1}\varclubsuit\ldots\varclubsuit v_{l_n},
v^{k_1}\otimes\ldots\otimes v^{k_n} \rangle,\qquad
\text{for } \mathbf{k}\le_{+}\mathbf{l}, \\
R^{\mathbf{k},\mathbf{l}}(q^{-1})
&:=&\langle v_{k_1}\otimes\ldots\otimes v_{k_n},
v^{l_1}\varspadesuit\ldots\varspadesuit v^{l_n}\rangle, \qquad
\text{for } \mathbf{k}\le_{-}\mathbf{l}.
\end{eqnarray*}
Then, the expansion of 
$v_{l_1}\varclubsuit\ldots\varclubsuit v_{l_n}$
in terms of $v_{k_1}\otimes\ldots\otimes v_{k_n}$
and its inverse is explicitly given by
\begin{eqnarray*}
v_{l_1}\varclubsuit\ldots\varclubsuit v_{l_n}
&=&
\sum_{\mathbf{k}\le_{+}\mathbf{l}}
R_{\mathbf{k},\mathbf{l}}(q^{-1})
v_{k_1}\otimes\ldots\otimes v_{k_n}, \\
v_{l_1}\otimes\ldots\otimes v_{l_n}
&=&\sum_{\mathbf{k}\le_{+}\mathbf{l}}
R^{\mathbf{l},\mathbf{k}}(q^{-1})
v_{k_1}\varclubsuit\ldots\varclubsuit v_{k_n}.
\end{eqnarray*}
We have similar formulas for 
$v^{k_1}\varspadesuit\ldots\varspadesuit v^{k_n}$.
Thus, the knowledge of $R_{\mathbf{k},\mathbf{l}}(q^{-1})$ and 
$R^{\mathbf{k},\mathbf{l}}(q^{-1})$ determines the expansion 
of dual canonical basis in terms of standard basis and 
vice versa. 
We will consider $R_{\mathbf{k},\mathbf{l}}(q^{-1})$ and 
$R^{\mathbf{k},\mathbf{l}}(q^{-1})$ in the next section.

From Eqn.(\ref{IProductSB}), we have an inversion formula:
\begin{eqnarray}
\label{InvR}
\sum_{\mathbf{k}}R_{\mathbf{k},\mathbf{l}}(q^{-1})
R^{\mathbf{k},\mathbf{l}}(q^{-1})
=\delta_{l_1}^{l'_1}\ldots\delta_{l_n}^{l'_n}.
\end{eqnarray}

\section{Path representation}
\label{Sec:Comb}

\subsection{\texorpdfstring{Maps $\zeta$ and $\eta$}{Maps zeta and eta}}
Let $l$ be an integer satisfying $0\le l\le m$.
We define a map $\zeta'$ from an integer $(m-2l)$ to a path 
of length $m$ by
\begin{eqnarray*}
\zeta'(m-2l)=\underbrace{+\ldots+}_{l}\underbrace{-\ldots-}_{m-l}.
\end{eqnarray*}
A map $\zeta$ from a sequence of integers 
$\mathbf{k}=(k_1,\ldots,k_n)\in I_{\mathbf{m}}$ to a path of length 
$\sum_{i=1}^{n}m_i$ is given by 
\begin{eqnarray}
\zeta(\mathbf{k})=\zeta'(k_1)\circ\ldots\circ\zeta'(k_n).
\end{eqnarray}
Here, $\zeta^{'}(k)\circ\zeta^{'}(k')$ is concatenation of 
two paths $\zeta^{'}(k)$ and $\zeta^{'}(k')$.

Similarly, a map $\eta'$ from an integer $(m-2k)$ to a path 
of length $m$ is defined by 
\begin{eqnarray*}
\eta(m-2k)=\underbrace{-\ldots-}_{m-k}\underbrace{+\ldots+}_{k}.
\end{eqnarray*}
We define a map $\eta$ from a sequence of integers 
$\mathbf{k}=(k_1,\ldots,k_n)\in I_{\mathbf{m}}$ to a path of length 
$\sum_{i=1}^n m_i$ by 
\begin{eqnarray*}
\eta(\mathbf{k})=\eta(k_1)\circ\ldots\circ\eta(k_n).
\end{eqnarray*}

Given $\mathbf{k}=(k_1,\ldots,k_n)$ and $\mathbf{l}=(l_1,\ldots,l_n)$,
the condition $\mathbf{k}\le_{-}\mathbf{l}$ is compatible with 
$\eta(\mathbf{k})\le\eta(\mathbf{l})$ at $\epsilon=-$, {\it i.e.},  
$\eta(\mathbf{k})$ is below $\eta(\mathbf{l})$.
Similarly, $\mathbf{k}\le_{+}\mathbf{l}$ if and only if 
$\eta(\mathbf{k})\le\eta(\mathbf{l})$ at $\epsilon=+$, {\it i.e.}, 
$\eta(\mathbf{k})$ is above $\eta(\mathbf{l})$.

\subsection{Ballot strips}
We make use of notations and terminologies for ballot strips 
used in~\cite[Section 3.1]{Shi14}.
We recollect necessary definitions since the rules of stacking ballot strips 
are slightly different from~\cite{Shi14}.

Let $\alpha\in\mathcal{P}_N$. We denote by $S(\alpha)$ the integral points
\begin{eqnarray*}
S(\alpha):=\{(i,j): (i,j)\text{ is above the path } \alpha, 0<i\le N,|j|<i, 
i+j-1\in2\mathbb{Z} \}.
\end{eqnarray*}
We put (45 degree rotated) squares of length $\sqrt{2}$ whose centers are all 
points in $S(\alpha)$. 
We regard this set of squares as a (45 degree rotated) shifted Ferrers diagram 
and denote it by $\lambda(\alpha)$.
We denote by $|\alpha|$ the number of squares in the shifted Ferrers diagram 
$\lambda(\alpha)$.
Given two paths $\alpha,\beta\in\mathcal{P}_N$ such that $\alpha$ is below 
$\beta$, we regard $\lambda(\alpha)/\lambda(\beta)$ as a 
shifted skew Ferrers diagram.
The number of squares in the skew diagram is denoted by 
$|\alpha|-|\beta|$.

Let $\mathbf{k}=(k_1,\ldots,k_n)\in I_{\mathbf{m}}$. 
Let $\alpha=\eta(\mathbf{k})$ and $\alpha'=\zeta(\mathbf{k})$ be two 
paths for $\mathbf{k}$. 
Let $\gamma$ be a path inbetween $\alpha$ and $\alpha'$.
We call the shifted skew Ferrers diagram $\lambda(\alpha)/\lambda(\gamma)$ 
a \emph{projection domain} ($p$-domain in short) for $\alpha$.
When we choose $\gamma=\alpha$ or all $m_i=1$, there is no $p$-domain. 

Two paths $\alpha=\eta(\mathbf{k})$ and $\beta=\eta(\mathbf{k'})$ 
($\alpha$ is below $\beta$) characterize the domains, the shifted skew 
Ferrers diagram $\mu:=\lambda(\alpha)/\lambda(\beta)$.
Let $P$ be a $p$-domain for $\alpha$. 
We call the domain $P\cap\mu$ a $p$-domain for $(\alpha,\beta)$.

\begin{defn}
We call a square whose centre is $(N,j)$ with $|j|<N$ and $N-j-1\in2\mathbb{Z}$ 
an \emph{anchor square}.
\end{defn}

We put ballot strips in a shifted skew Ferrers diagram.
We have one constraint for a ballot strip: 
the rightmost square of a ballot strip of length $(l,l')$ with $l'\ge1$ is on an 
anchor square in the shifted skew Ferrers diagram. 

Let $\mathcal{D},\mathcal{D'}$ be two ballot strips. 
We pile $\mathcal{D'}$ on top of $\mathcal{D}$ by the following two rules 
in a shifted skew Ferrers diagram $\lambda(\alpha)/\lambda(\beta)$:
\begin{enumerate}[Rule I]
\item
\begin{enumerate}
\item If there exists a square of $\mathcal{D}$ just below a square of $\mathcal{D'}$, 
then all squares just below a square of $\mathcal{D'}$ belong to $\mathcal{D}$.
\item Suppose $l'\ge1$. 
The number of ballot strips of length $(l,l')$ is even
for $l'-1\in2\mathbb{Z}$, and zero for otherwise. 
\item There is no $p$-domain for $(\alpha,\beta)$ in the shifted skew Ferrers 
diagram.
\end{enumerate}
\item
\begin{enumerate}
\item If there exists a square of $\mathcal{D'}$ just above, NW, or NE of a square of
$\mathcal{D}$, then all squares just above, NW, and NE of a square of $\mathcal{D}$ 
belong to $\mathcal{D}$ or $\mathcal{D'}$ except in a $p$-domain.
\item Suppose $l'\ge1$. 
If there exists a ballot strip $\mathcal{D}$ of length $(l,l')$ with
$l'-1\in2\mathbb{Z}$ (resp. $l'\in2\mathbb{Z}$), 
then there is a strip of length $(l'',l'+1)$ with $l''\ge l$ 
(resp. $(l'',l'-1)$ with $l''\le l$) just above
(resp. just below) $\mathcal{D}$.
\item There is no squares of a strip $\mathcal{D}$ of length $(l,l')$ with $l+l'\ge1$
in a $p$-domain. 
\item If a $p$-domain for $(\alpha,\beta)$ exists, fill it with unit squares.
\end{enumerate}
\end{enumerate}

Let $\alpha<\beta$ be two paths.
We denote by $\mathrm{Conf}(\alpha,\beta)$ the set of all possible 
configurations of ballot strips, and by $\mathrm{Conf}^{I/II}(\alpha,\beta)$ 
the subset of configurations satisfying Rule I/II.
In the case of Rule II, there may be several choices of a $p$-domain 
for $(\alpha,\beta)$. 
However, given a $p$-domain for $(\alpha,\beta)$, 
there exists at most one configuration satisfying Rule II. 

The weight of a ballot strip is defined by
\begin{eqnarray*}
\mathrm{wt}^{I}(\mathcal{D})
&:=&
\left\{
\begin{array}{cc}
t^{-1}, & l'\text{ is even}, \\
1, & l' \text{ is odd},
\end{array}
\right. \\
\mathrm{wt}^{II}(\mathcal{D})
&:=&
\left\{
\begin{array}{cc}
-t^{-1}, & l'\text{ is even}, \\
1, & l' \text{ is odd}, \\
t^{-1} & \text{a unit square in a $p$-domain}.
\end{array}
\right.
\end{eqnarray*}

\begin{defn}
The generating function of ballot strips for the paths 
$\alpha<\beta$ is defined by
\begin{eqnarray*}
Q_{\alpha,\beta}^{X,\epsilon}
=
\sum_{\mathcal{C}\in\mathrm{Conf}^{X}(\alpha,\beta)}
\prod_{\mathcal{D}\in\mathcal{C}}
\mathrm{wt}^{X}(\mathcal{D}),
\end{eqnarray*}
where $X={I,II}$ and $\epsilon=\pm$.
The order $\alpha<\beta$ is associated with the sign $\epsilon$. 
We define $Q_{\alpha,\alpha}^{X,\epsilon}=1$.
\end{defn}

\begin{example}
Let $\alpha=--++-+$ and $\beta=++++++$. 
The possible configurations satisfying Rule I are
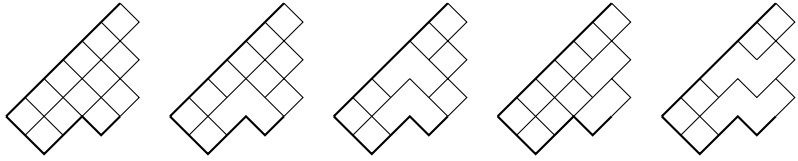
\begin{eqnarray*}
\begin{tikzpicture}[scale=0.25]
\draw[thick](0,0)--(2,-2)--(4,0)--(5,-1)--(6,0)(0,0)--(6,6);
\draw(1,-1)--(7,5)--(6,6)(4,0)--(7,3)--(5,5)(6,0)--(7,1)--(4,4);
\draw(1,1)--(3,-1)(2,2)--(4,0)(3,3)--(6,0);
\end{tikzpicture}\quad
\begin{tikzpicture}[scale=0.25]
\draw[thick](0,0)--(2,-2)--(4,0)--(5,-1)--(6,0)(0,0)--(6,6);
\draw(1,-1)--(7,5)--(6,6)(5,1)--(7,3)--(5,5)(6,0)--(7,1)--(4,4);
\draw(1,1)--(3,-1)(2,2)--(3,1)(3,3)--(6,0);
\end{tikzpicture}\quad
\begin{tikzpicture}[scale=0.25]
\draw[thick](0,0)--(2,-2)--(4,0)--(5,-1)--(6,0)(0,0)--(6,6);
\draw(1,-1)--(4,2)(5,3)--(7,5)--(6,6)(5,1)--(7,3)--(5,5)(6,0)--(7,1)--(4,4);
\draw(1,1)--(3,-1)(2,2)--(3,1)(4,2)--(6,0);
\end{tikzpicture}\quad
\begin{tikzpicture}[scale=0.25]
\draw[thick](0,0)--(2,-2)--(4,0)--(5,-1)--(6,0)(0,0)--(6,6);
\draw(1,-1)--(7,5)--(6,6)(2,-2)--(7,3)--(5,5)(6,0)--(7,1)--(6,2);
\draw(1,1)--(3,-1)(2,2)--(4,0)(3,3)--(5,1)(4,4)--(5,3);
\end{tikzpicture}\quad
\begin{tikzpicture}[scale=0.25]
\draw[thick](0,0)--(2,-2)--(4,0)--(5,-1)--(6,0)(0,0)--(6,6);
\draw(1,-1)--(4,2)--(5,1)--(7,3)--(5,5)(5,3)--(7,5)--(6,6);
\draw(1,1)--(3,-1)(2,2)--(3,1)(4,4)--(5,3)(6,0)--(7,1)--(6,2);
\end{tikzpicture}
\end{eqnarray*}
The generating function is 
$Q_{\alpha,\beta}^{I,-}=t^{-13}+t^{-11}+2t^{-9}+t^{-5}$.
\end{example}

\begin{example}
Let $\mathbf{m}=(2,2,2,2), \mathbf{k}=(0,0,0,0)$ and  $\mathbf{k'}=(-2,-2,-2,2)$.
Then, $\alpha:=\eta(\mathbf{k})=-+-+-+-+$ and $\beta:=\eta(\mathbf{k'})=++++++--$. 
The possible configurations satisfying Rule II are
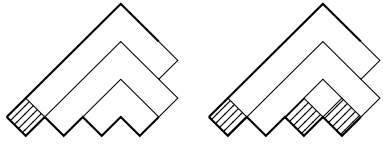
\begin{eqnarray*}
\begin{tikzpicture}[scale=0.25]
\draw[thick](0,0)--(6,6)--(8,4)
(0,0)--(1,-1)--(2,0)--(3,-1)--(4,0)--(5,-1)--(6,0)--(7,-1)--(8,0);
\draw(2,0)--(6,4)--(9,1)--(8,0)(1,1)--(2,0)(8,4)--(9,3)--(8,2);
\draw(4,0)--(6,2)--(8,0);
\fill[pattern=north east lines](0,0)--(1,1)--(2,0)--(1,-1)--(0,0);
\end{tikzpicture}\quad
\begin{tikzpicture}[scale=0.25]
\draw[thick](0,0)--(6,6)--(8,4)
(0,0)--(1,-1)--(2,0)--(3,-1)--(4,0)--(5,-1)--(6,0)--(7,-1)--(8,0);
\draw(2,0)--(6,4)--(9,1)--(8,0)(1,1)--(2,0)(8,4)--(9,3)--(8,2);
\draw(4,0)--(6,2)--(8,0)(5,1)--(6,0)--(7,1);
\fill[pattern=north east lines](0,0)--(1,1)--(2,0)--(1,-1)--(0,0);
\fill[pattern=north east lines](4,0)--(5,1)--(6,0)--(5,-1)--(4,0)
(6,0)--(7,1)--(8,0)--(7,-1)--(6,0);
\end{tikzpicture}
\end{eqnarray*}
where a square  
\raisebox{-0.3\totalheight}{
\begin{tikzpicture}[scale=0.25]
\fill[pattern=north east lines](0,0)--(1,1)--(2,0)--(1,-1)--(0,0);
\draw(0,0)--(1,1)--(2,0)--(1,-1)--(0,0);
\end{tikzpicture}}
is in a $p$-domain.
The generating function is $Q_{\alpha,\beta}^{II,-}=t^{-3}+t^{-5}$.
\end{example}

From Lemma 4 in~\cite{SZJ12}, we have 
\begin{lemma}
\label{transposeQ}
\begin{eqnarray}
Q_{\alpha,\beta}^{X,+}(t^{-1})
=
Q_{\beta,\alpha}^{X,-}(t^{-1}).
\end{eqnarray}
\end{lemma}

From Theorem~21 in \cite{Shi14}, we have the following two lemmas: 
\begin{lemma}\label{lemma-QP1}
\begin{eqnarray}
Q_{\alpha,\beta}^{I,+}(t^{-1})
=
P_{\alpha,\beta}^{+}(t^{-1}).
\end{eqnarray}
\end{lemma}

For given two paths $\alpha$ and $\beta$, we define a path $\gamma$ by
the following two conditions: 1) $\gamma$ is below both $\alpha$ and $\beta$, and 
2) there is no path above $\gamma$ satisfying 1). 
We denote the path $\gamma$ by $\min(\alpha,\beta)$.

For example, $\min(+---+,-+++-)=-+--+$.

\begin{lemma}
\label{lemma-QP2}
Let $\alpha=\eta(\mathbf{k}), \beta=\eta(\mathbf{k'})$
and $\alpha'=\min(\zeta(\mathbf{k}),\beta)$. 
\begin{eqnarray}
Q_{\alpha,\beta}^{II,-}(t^{-1})
=
\sum_{\alpha\le\gamma\le\alpha'}t^{-|\alpha|+|\gamma|}
P_{\gamma,\beta}^{-}(-t^{-1}).
\end{eqnarray}
\end{lemma}
\begin{proof}
Given $\gamma$ ($\alpha\le\gamma\le\alpha'$), the domain 
inbetween $\alpha$ and $\gamma$ is a $p$-domain for $(\alpha,\beta)$.
The number of squares in a $p$-domain is $|\alpha|-|\gamma|$.
This gives the factor $t^{-|\alpha|+|\gamma|}$.

The domain inbetween $\gamma$ and $\beta$ is filled with ballot strips
satisfying Rule II-(a) and II-(b).
Since these two rules are the same as the ones for Kazhdan--Lusztig 
polynomials (Section 3.1 in~\cite{Shi14}), the contribution is nothing but 
the Kazhdan--Lusztig polynomial $P^{-}_{\gamma,\beta}(-t^{-1})$ 
(Corollary 20 in~\cite{Shi14}). 
The sum over all possible $\gamma$ gives the desired equation.
\end{proof}

\subsection{\texorpdfstring{Path representation of $R$}{Path representation of R}}

\begin{theorem}
Let $\mathbf{k,l}\in I_{\mathbf{m}}$, $\alpha=\eta(\mathbf{k})$ and 
$\beta=\eta(\mathbf{l})$.
\begin{eqnarray}
R^{\mathbf{k},\mathbf{l}}(q^{-1})=Q_{\alpha,\beta}^{II,-}(q^{-1}).
\end{eqnarray}
\end{theorem}
\begin{proof}

From Lemma~\ref{lemma-KL-dC}, the dual canonical basis in 
$V_1^{\otimes N}$ is expanded as 
\begin{eqnarray}
\label{eq-dC-KL}
v^{\kappa_1}\varspadesuit\ldots\varspadesuit v^{\kappa_N}
=
\sum_{\gamma_1\le\gamma_2}
P_{\gamma_1,\gamma_2}^{-}(-q^{-1})
v^{\kappa'_1}\otimes\ldots\otimes v^{\kappa'_N}
\end{eqnarray}
where $\gamma_1=\eta(\mathbf{\kappa'})$ and $\gamma_2=\eta(\mathbf{\kappa})$.
Given $\mathbf{\kappa'}$, we define $k_i, 1\le i\le n$, as 
$k_i=\kappa'_{d_{i-1}+1}+\ldots+\kappa'_{d_{i}}$,
where $d_{i}:=\sum_{l=1}^{i}m_{l}$ and $d_0=0$.
The action of $\pi=\pi_1\otimes\ldots\otimes\pi_n$ on 
$v^{\kappa'_1}\otimes\ldots\otimes v^{\kappa'_N}$
is 
\begin{eqnarray*}
\pi(v^{\kappa'_1}\otimes\ldots\otimes v^{\kappa'_N})
=
q^{-D}v^{k_1}\otimes\ldots\otimes v^{k_n}
\end{eqnarray*}
where $D=|\eta(\mathbf{k})|-|\eta(\mathbf{\kappa'})|$.
 
From Theorem~\ref{Thm-dC}, the dual canonical basis is obtained 
by multiplying $\pi$ on 
$v^{\kappa_1}\varspadesuit\ldots\varspadesuit v^{\kappa_N}$ 
for an appropriate choice of $\mathbf{\kappa}$.
For such $\kappa$, the multiplication of $\pi$ on 
Eqn.(\ref{eq-dC-KL}) yields
\begin{eqnarray*}
v^{l_1}\varspadesuit\ldots\varspadesuit v^{l_N}
=
\sum_{\alpha\le\beta}
\sum_{\alpha\le\gamma\le\alpha'}
q^{-|\alpha|+|\gamma|}P_{\gamma,\beta}^{-}(-q^{-1})
v^{k_1}\otimes\ldots\otimes v^{k_n}
\end{eqnarray*}
where $\alpha=\eta(\mathbf{k}), \beta=\eta(\mathbf{l})$
and $\alpha'=\min(\alpha,\beta)$.
Together with Lemma~\ref{lemma-QP2}, we complete the proof.
\end{proof}

\subsection{Inversion formula}
Fix $\mathbf{m}$. 
Suppose that a path $\alpha$ is below a path $\beta$ and 
written as $\alpha=\eta(\mathbf{k})$ and 
$\beta=\eta(\mathbf{k'})$ with $\mathbf{k,k'}\in I_{\mathbf{m}}$.
A path $\gamma$ 
is said to be {\it admissible} if and only if there exists 
$\mathbf{k''}\in I_{\mathbf{m}}$ satisfying $\gamma=\eta(\mathbf{k''})$ 
and $\gamma$ is above $\alpha$ and below $\beta$.
We denote by $\mathrm{Adm}(\alpha,\beta)$ the set of all admissible 
paths inbetween two paths $\alpha$ and $\beta$.
\begin{theorem}
\label{InvQ}
\begin{eqnarray}
\sum_{\beta\in\mathrm{Adm}(\alpha,\gamma)}
Q_{\alpha,\beta}^{I,-}(t^{-1})Q_{\beta,\gamma}^{II,-}(t^{-1})=\delta_{\alpha,\gamma}.
\end{eqnarray}
\end{theorem}
\begin{proof}
We prove the theorem along~\cite[Theorem 6]{SZJ12} and \cite[Theorem 10]{Shi14}.
The main differences are the existence of $p$-domains and admissibility of a path.

When $\alpha=\gamma$, Theorem holds true. 
Hereafter, we assume $\alpha\neq\gamma$.
Fix two paths $\alpha,\gamma$ and a configuration 
$\mathcal{C}\in\mathrm{Conf}(\alpha,\gamma)$ such that there exists a path 
$\beta$ satisfying $\mathcal{C}=\mathcal{C}_{I}(\beta)\sqcup\mathcal{C}_{II}(\beta)$
where $\mathcal{C}_{I}(\beta)\in\mathrm{Conf}^{I}(\alpha,\beta)$ and 
$\mathcal{C}_{II}(\beta)\in\mathrm{Conf}^{II}(\beta,\gamma)$.
We denote by $P(\mathcal{C})$ the set of such paths $\beta$'s.
If there exist $p$-domains in $\mathcal{C}_{II}(\beta)$, we define a path 
$\bar{\beta}$ such that the shifted skew Ferrers diagram 
$\lambda(\beta)/\lambda(\bar{\beta})$ is the $p$-domains for $(\beta,\gamma)$.
If there is no $p$-domain, we define $\bar{\beta}=\beta$.
We denote by $\bar{P}(\mathcal{C})$ the set of such $\bar{\beta}$.
Notice that $\bar{\beta}$ may not be admissible.
Reversely, one can obtain an admissible path $\beta$ from $\bar{\beta}$ as follows.
If $\bar{\beta}$ is admissible, then $\beta=\bar{\beta}$.
If a path $\bar{\beta}$ is not admissible, 
one can obtain an admissible path $\beta$ below $\bar{\beta}$ such that 
the number of squares filling the region $\lambda(\beta)/\lambda(\bar{\beta})$ 
is minimum.  
Therefore, there is a one-to-one correspondence between an admissible path 
in $P(\mathcal{C})$ and a path in $\bar{P}(\mathcal{C})$.
Since a unit square can be piled on top of any other ballot strip by Rule I,
the region inbetween $\alpha$ and $\bar{\beta}$ satisfies Rule I.
The configuration $\mathcal{C}$ is written as 
$\mathcal{C}=\mathcal{C}_{I}(\bar{\beta})\sqcup\mathcal{C}_{II}(\bar{\beta})$.
We define the set of ballot strips $\mathcal{I}(C)$ as
\begin{eqnarray*}
\mathcal{I}(C):=
\left(
\bigcup_{\bar{\beta}\in \bar{P}(\mathcal{C})}
\mathcal{C}_{I}(\bar{\beta})
\right)
\cap
\left(
\bigcup_{\bar{\beta}\in \bar{P}(\mathcal{C})}
\mathcal{C}_{II}(\bar{\beta})
\right).
\end{eqnarray*}
An element in $\mathcal{I}(C)$ is a single ballot strip or ballot strips which 
are on top of each other and glued together.
The element of $\mathcal{I}(C)$ is on the border of $C_{I}(\bar{\beta})$ and 
$C_{II}(\bar{\beta})$.
By a similar argument to Proposition~4 in \cite{SZJ12}, one can 
show $\mathcal{I}(C)\neq\emptyset$.
Further, Lemma~5 in \cite{SZJ12} holds true as well. 
From the definition of $\mathcal{I}(C)$, a path $\bar{\beta}\in\bar{P}(C)$ is 
characterized by being above or below an element in $\mathcal{I}(C)$.
Thus, the cardinality of $\bar{P}(C)$ is $2^{|\mathcal{I}(C)|}$.

We denote by $\mathrm{wt}(\mathcal{C})$ the weight of a configuration 
$\mathcal{C}$. 
The weight of a unit square in a $p$-domain for Rule II is $t^{-1}$, which is 
the same value as the weight of a unit square for Rule I.
We have 
\begin{eqnarray*}
\sum_{\beta\in\mathrm{Adm}(\alpha,\gamma)}
Q_{\alpha,\beta}^{I,-}Q_{\beta,\gamma}^{II,-}
&=&\sum_{\mathcal{C}}\sum_{\beta\in P(\mathcal{C})}\mathrm{wt}(\mathcal{C}) \\
&=&\sum_{\mathcal{C}}|\mathrm{wt}(\mathcal{C})|\sum_{\bar{\beta}\in \bar{P}(\mathcal{C})}
\mathrm{sign}(\mathcal{C}).
\end{eqnarray*}
To prove the theorem, it is enough to show 
$\sum_{\bar{\beta}\in \bar{P}(\mathcal{C})}\mathrm{sign}(\mathcal{C})=0$.
We first show this equality in the case of two paths.
Let $\beta_1,\beta_2\in\bar{P}(\mathcal{C})$ be two paths such that 
there exists an element $\mathcal{D}$ 
belonging to $\mathcal{C}_{I}(\beta_1)$ and $\mathcal{C}_{II}(\beta_2)$. 
In other words, this means $\beta_1$ is above $\beta_2$.
We have two cases for $\mathcal{D}$:
\begin{enumerate}[{Case} 1.]
\item $\mathcal{D}$ is of length $(l,0)$.
The sign of $\mathcal{D}$ for Rule I is plus, whereas minus for Rule II.
Thus the contributions cancel each other.
\item $\mathcal{D}$ is glued ballot strips.
For Rule I, $\mathcal{D}$ consists of three elements: two ballot strips of 
same length $(l,l')$ with $l'-1\in2\mathbb{Z}$ and a ballot strip of 
length $(l'',0)$ for some $l''$.
On the other hand, for Rule II, $\mathcal{D}$ consists of the ballot strip of length 
$(l,l')$ and a ballot strip of length $(l+l'',l'+1)$.
The sign is plus for Rule I and minus for Rule II.
Thus the contributions cancel each other.
\end{enumerate}
Finally, $\sum_{\bar{\beta}\in \bar{P}(\mathcal{C})}\mathrm{sign}(\mathcal{C})$
is reduced to the sum of contributions by local sum, which is zero.
This completes the proof.
\end{proof}

\begin{theorem}
We have 
\begin{eqnarray}
R_{\mathbf{k},\mathbf{l}}(q^{-1})
=
Q_{\eta(\mathbf{k}),\eta(\mathbf{l})}^{I,+}(q^{-1})
\end{eqnarray}
\end{theorem}
\begin{proof}
From inversion relations (\ref{InvR}) and Theorem~\ref{InvQ} together with 
Lemma~\ref{transposeQ} , we have 
$R_{\mathbf{k},\mathbf{l}}
=Q_{\eta(\mathbf{l}),\eta(\mathbf{k})}^{I,-} 
=Q_{\eta(\mathbf{k}),\eta(\mathbf{l})}^{I,+}$
\end{proof}
From Lemma~\ref{lemma-QP1}, it is obvious that 
\begin{cor}
\label{cor-RP}
\begin{eqnarray}
R_{\mathbf{k},\mathbf{l}}(q^{-1})
=
P_{\eta(\mathbf{k}),\eta(\mathbf{l})}^{+}(q^{-1}).
\end{eqnarray}
\end{cor}

\subsection{Canonical bases of U}
We first prove two lemmas used later. 
\begin{lemma}
\label{lemma-qbinomial}
Let $l\in I_{(m)}$ and $\epsilon=+$.
We have
\begin{eqnarray}
\label{qbinomial}
\sum_{\alpha'\le\gamma\le\alpha}q^{2|\gamma|-|\alpha|-|\alpha'|}
=\genfrac{[}{]}{0pt}{}{m}{(m-l)/2}
\end{eqnarray}
where $\alpha=\eta(l)$ and $\alpha'=\zeta(l)$.
\end{lemma}
\begin{proof} 
We prove the lemma by induction. When $m=1$ and $m=2$, 
the lemma is obviously true.
We assume that the lemma holds true for all $m'<m$ and all $l\in I_{(m')}$.
A path $\gamma$ starts with $+$ or $-$.
Let $P_+$ (resp. $P_-$) be the set of $\gamma$, $\alpha'\le\gamma\le\alpha$, 
starting $+$ (resp. $-$).
We define two paths: 
$\alpha_1:=+\underbrace{-\cdots-}_{(m+l)/2}\underbrace{+\cdots+}_{(m-l)/2-1}$ 
and 
$\alpha'_1:=-\underbrace{+\cdots+}_{(m-l)/2}\underbrace{-\cdots-}_{(m+l)/2-1}$.
Then we have 
\begin{eqnarray*}
\sum_{\gamma\in P_{+}}q^{2|\gamma|-|\alpha|-|\alpha'|}
&=&
q^{-(m+l)/2}\sum_{\gamma\in P_{+}}q^{2|\gamma|-|\alpha_1|-|\alpha'|} \\
&=&q^{-(m+l)/2}\genfrac{[}{]}{0pt}{}{m-1}{(m-l)/2-1}
\end{eqnarray*}
and 
\begin{eqnarray*}
\sum_{\gamma\in P_{-}}q^{2|\gamma|-|\alpha|-|\alpha'|}
&=&
q^{(m-l)/2}\sum_{\gamma\in P_{-}}q^{2|\gamma|-|\alpha|-|\alpha'_1|} \\
&=&q^{(m-l)/2}\genfrac{[}{]}{0pt}{}{m-1}{(m-l)/2}
\end{eqnarray*} 
where we have used the induction assumption. 
The sum of two contributions is the right hand side of Eqn.(\ref{qbinomial}).
\end{proof}

\begin{lemma}
\label{lemma-PP}
Let $\alpha=\eta(\mathbf{l})$, $\alpha'=\zeta(\mathbf{l})$ and 
$\beta=\eta(\mathbf{k})$ with $\mathbf{k,l}\in I_{\mathbf{m}}$. 
We have 
\begin{eqnarray}
\label{PP-relation}
P^+_{\alpha,\beta}(q^{-1})
=
\sum_{\alpha'\le\gamma\le\alpha}q^{|\gamma|-|\alpha'|}
\prod_{i=1}^{n}\genfrac{[}{]}{0pt}{}{m_i}{(m_i-l_i)/2}^{-1}	
P^{+}_{\gamma,\beta}(q^{-1}).
\end{eqnarray}
\end{lemma}
\begin{proof}
We first show that if $\alpha'\le\gamma\le\alpha$ ($\epsilon=+$), 
$P^{+}_{\gamma,\beta}=q^{-|\alpha|+|\gamma|}P^{+}_{\alpha,\beta}$.
Recall that $P^{+}_{\alpha,\beta}=Q^{I,+}_{\alpha,\beta}$ 
(Lemma~\ref{lemma-QP1}). 
The region between $\alpha$ and $\beta$ is filled with ballot strips 
by Rule I. 
One can pile a ballot strip on top of a ballot strip of length $(l,l')$ 
with $l+l'\ge1$ if and only if 
there is a convex path $+-$ below the strip.
We consider the region between $\alpha'$ and $\beta$.
From the definition of $\eta$, there is no convex path below the 
region $\lambda(\alpha)/\lambda(\gamma)$.
Thus there is no ballot strips of length $(l,l')$ with $l+l'\ge1$ 
in this region. 
Since the weight of a unit square is $t^{-1}$, we have 
$P^{+}_{\gamma,\beta}=q^{-|\alpha|+|\gamma|}P^{+}_{\alpha,\beta}$.

The right hand side of Eqn.(\ref{PP-relation}) is 
\begin{eqnarray*}
\sum_{\alpha'\le\gamma\le\alpha}q^{|\gamma|-|\alpha'|}
\prod_{i=1}^{n}\genfrac{[}{]}{0pt}{}{m_i}{(m_i-l_i)/2}^{-1}	
P^{+}_{\gamma,\beta}(q^{-1})
&=&
\prod_{i=1}^{n}
\genfrac{[}{]}{0pt}{}{m_i}{(m_i-l_i)/2}^{-1}	
P^{+}_{\alpha,\beta}(q^{-1})
\sum_{\alpha'\le\gamma\le\alpha}q^{2|\gamma|-|\alpha|-|\alpha'|}  \\
&=&
P^{+}_{\alpha,\beta}
\end{eqnarray*}
where we have used Lemma~\ref{lemma-qbinomial}. 
This completes the proof.
\end{proof}

The following theorem is the dual statement of Theorem~\ref{Thm-dC}. 
Given $\mathbf{k}\in I_{\mathbf{m}}$, let $(\kappa_1,\ldots,\kappa_N)$ 
be the sequence defined just above Theorem~\ref{Thm-dC}.
\begin{theorem}
The canonical basis $v_{m_1-2k_1}\varclubsuit\cdots\varclubsuit v_{m_n-2k_n}$
is given by  
\begin{eqnarray}
\label{projection-CB}
v_{m_1-2k_1}\varclubsuit\cdots\varclubsuit v_{m_n-2k_n}
=(\pi_1\otimes\cdots\otimes\pi_n)
v_{\kappa_1}\varclubsuit\cdots\varclubsuit v_{\kappa_N}.
\end{eqnarray}
\end{theorem}
\begin{proof}
We compute the right hand side of Eqn.(\ref{projection-CB}) briefly.
We expand $v_{\kappa_1}\varclubsuit \cdots\varclubsuit v_{\kappa_N}$ 
in terms of the standard basis $v_{\kappa'_1}\otimes\cdots\otimes v_{\kappa'_N}$.
We also expand $v_{m_1-2k_1}\varclubsuit \cdots\varclubsuit v_{m_n-2k_n}$ 
in terms of the standard basis $v_{m_1-2l_1}\otimes\cdots\otimes v_{m_n-2l_n}$. 
Let $\alpha=\eta(m_1-l_1,\ldots,m_n-l_n)$, $\alpha'=\zeta(m_1-l_1,\ldots,m_n-l_n)$,   
$\beta=\eta(\kappa)$ and $\gamma=\eta(\kappa')$. 
The coefficient is nothing but the Kazhdan--Lusztig polynomials $P^{+}_{\gamma,\beta}$. 
The projection (\ref{projection-dual}) gives the factor 
$q^{|\gamma|-|\alpha'|}\prod_{i=1}^{n}\genfrac{[}{]}{0pt}{}{m_i}{(m_i-l_i)/2}^{-1}$.  	
Thus the sum of all the contributions is reduced to the right hand side of 
Eqn.(\ref{PP-relation}). 
From Lemma~\ref{lemma-PP} and Corollary~\ref{cor-RP}, we obtain 
Eqn.(\ref{projection-CB}).
\end{proof}

\section{Integral structure}
\label{Sec:IS}
\subsection{Integral structure of tensor products}
From the definitions of an arc and a dashed arc, we have the following lemma.
\begin{lemma}
\label{moveuparrow}
We have two identities:
\begin{eqnarray}
\label{moveuparroweq1}
\raisebox{-0.3\totalheight}{
\begin{tikzpicture}
\draw(-0.8,0)--(-0.8,-0.7)(-0.87,-0.5)--(-0.8,-0.7)--(-0.73,-0.5);
\draw(0,0)--(0,-0.7)(-0.07,-0.2)--(0,0)--(0.07,-0.2);
\end{tikzpicture}}
&=&
\raisebox{-0.5\totalheight}{
\begin{tikzpicture}
\draw(0,0)..controls(0,-1)and(0.8,-1)..(0.8,0);
\end{tikzpicture}}
\quad+q^{-1}\ 
\raisebox{-0.3\totalheight}{
\begin{tikzpicture}
\draw(-0.8,0)--(-0.8,-0.7)(-0.87,-0.2)--(-0.8,0)--(-0.73,-0.2);
\draw(0,0)--(0,-0.7)(-0.07,-0.5)--(0,-0.7)--(0.07,-0.5);
\end{tikzpicture}}, \\
\label{moveuparroweq2}
\raisebox{-0.3\totalheight}{
\begin{tikzpicture}
\draw(-0.8,0)--(-0.8,-0.7)(-0.87,-0.5)--(-0.8,-0.7)--(-0.73,-0.5);
\draw(0,0)--(0,-0.7)(-0.07,-0.5)--(0,-0.7)--(0.07,-0.5);
\end{tikzpicture}}
&=&
\raisebox{-0.5\totalheight}{
\begin{tikzpicture}
\draw[dashed,thick](0,0)..controls(0,-1)and(0.8,-1)..(0.8,0);
\end{tikzpicture}}
\quad
+q^{-1}\ 
\raisebox{-0.3\totalheight}{
\begin{tikzpicture}
\draw(-0.8,0)--(-0.8,-0.7)(-0.87,-0.2)--(-0.8,0)--(-0.73,-0.2);
\draw(0,0)--(0,-0.7)(-0.07,-0.2)--(0,0)--(0.07,-0.2);
\end{tikzpicture}}.
\end{eqnarray}
\end{lemma}

\begin{theorem}
\label{thm-decomp-dual-dual}
The coefficients of the decomposition
\begin{eqnarray*}
v^{l_1}\heartsuit\cdots\heartsuit v^{l_N}
=
\sum_{\mathbf{k}}c^{\mathbf{l}}_{\mathbf{k}}
v^{k_1}\varspadesuit\cdots\varspadesuit v^{k_N}
\end{eqnarray*}
belong to $q^{-1}\mathbb{N}[q^{-1}]$ except 
$c_{\mathbf{l}}^{\mathbf{l}}=1$.
\end{theorem}
\begin{proof}
The diagram $D$ of $v^{l_1}\heartsuit\cdots\heartsuit v^{l_N}$ consists
of arcs, up arrows and down arrows.
The down arrows are right to the up arrows in $D$. 
To obtain the expansion of $D$ in terms of 
$D':=v^{k_1}\varspadesuit\cdots\varspadesuit v^{k_N}$, 
we successively use Lemma~\ref{moveuparrow}.
We change two adjacent down arrows into a dashed arc, and move 
up arrows to the left of down arrows.
The coefficients in the right hand sides of Eqns.(\ref{moveuparroweq1}) 
and (\ref{moveuparroweq2})
are one for an arc and a dashed arc, and are in $q^{-1}\mathbb{N}[q^{-1}]$ 
for otherwise.
Thus, the coefficient $c_{\mathbf{l}}^{\mathbf{l}}=1$ and other coefficients 
are in $q^{-1}\mathbb{N}[q^{-1}]$.
This completes the proof.
\end{proof}

\begin{remark}
\label{rem-dec}
In the proof of Theorem~\ref{thm-decomp-dual-dual}, 
the numbers of arcs and up arrows in $D'$ may be more than those in $D$.  
However, if $D$ has an arc (resp. an up arrow), then $D'$ has also an arc 
(resp. an up arrow) in the same position. 
The coefficients of decomposition of 
$v^{l_1}\heartsuit\cdots\heartsuit v^{l_N}$
is equivalent to the decomposition of $v^{-1}\otimes\cdots\otimes v^{-1}$ where 
the number of $v^{-1}$ is the same as the number of down arrows in $D$.
This implies that 
$c_{\mathbf{l}}^{\mathbf{k}}=R_{\mathbf{k}_0,\mathbf{l}_0}$ for some $\mathbf{k}_0$ 
and $\mathbf{l}_0$.
\end{remark}

\begin{theorem}
\label{thm-dec-dual-dual2}
The coefficients of the decomposition
\begin{eqnarray*}
v^{l_1}\varspadesuit\cdots\varspadesuit v^{l_N}
=
\sum_{\mathbf{k}}c_{\mathbf{k}}^{\mathbf{l}}
v^{k_1}\heartsuit\cdots\heartsuit v^{k_N},
\end{eqnarray*}
belong to $q^{-1}\mathbb{Z}[q^{-1}]$ except 
$c_{\mathbf{l}}^{\mathbf{l}}=1$.
\end{theorem}
\begin{proof}
We expand a dashed arc and a down arrow with a star in terms of 
up and down arrows where the coefficients of up arrows are in
$q^{-1}\mathbb{Z}[q^{-1}]$.
We move up arrows to the left by using Eqn.(\ref{moveuparroweq1}) 
in Lemma~\ref{moveuparrow}.
We have $c_{\mathbf{l}}^{\mathbf{l}}=1$ and 
$c_{\mathbf{l}}^{\mathbf{k}}\in q^{-1}\mathbb{Z}[q^{-1}]$.
\end{proof}

\begin{remark}
By a similar argument to Remark~\ref{rem-dec} and 
Theorem~\ref{thm-dec-dual-dual2}, the decomposition 
of $v^{l_1}\varspadesuit\cdots\varspadesuit v^{l_N}$ 
is divided into two steps: (1) decompose 
$v^{-1}\varspadesuit\cdots\varspadesuit v^{-1}$ into 
$v^{\kappa_1}\otimes\cdots\otimes v^{\kappa_N'}$.
(2) Decomposition of $v^{\kappa_1}\otimes\cdots\otimes v^{\kappa_N'}$
into $v^{\sigma_1}\heartsuit\cdots\heartsuit v^{\sigma_N'}$.
The contribution of step (1) in $c_{\mathbf{k}}^{\mathbf{l}}$ is 
$R^{\mathbf{k}_0,\mathbf{l}_0}$ for some $\mathbf{k}_0$ and $\mathbf{l}_0$. 
Similarly, the contribution of step (2) is $P^{A,+}_{\mathbf{k}_0,\mathbf{l}_1}$ 
where $P^{A,+}_{\mathbf{k}_0,\mathbf{l}_1}$
is the Kazhdan--Lusztig polynomial of type $A$ (see, {\it e.g.}, \cite{SZJ12}). 
Thus,  
$c_{\mathbf{k}}^{\mathbf{l}}
=
R^{\mathbf{k}_0,\mathbf{l}_0}P^{A,+}_{\mathbf{k}_0,\mathbf{l}_1}$ 
for some $\mathbf{k}_0, \mathbf{l}_0$ and $\mathbf{l}_1$.
\end{remark}

Let $M$ (resp. $N$) be a product of finitely many irreducible representations
of $U_q(\mathfrak{sl}_2)$ (resp. $U$).
Let $\{b_{\mathbf{k}}|\mathbf{k}\in I_{\mathbf{m}}\}$ be dual canonical bases 
in $M$ and $\{b'_{\mathbf{k'}}|\mathbf{k'}\in I_{\mathbf{m'}}\}$ be dual 
canonical bases in $N$.
\begin{theorem}
Consider the decomposition 
\begin{eqnarray*}
b_{\mathbf{k}}\otimes b'_{\mathbf{k'}}
=
\sum_{\mathbf{l}\in I_{\mathbf{m}},\mathbf{l'}\in I_{\mathbf{m'}}}
c^{\mathbf{l,l'}}_{\mathbf{k,k'}}b_{\mathbf{l}}\varspadesuit b_{\mathbf{l'}}.
\end{eqnarray*}
where $b_{\mathbf{l}}\varspadesuit b_{\mathbf{l'}}$ is a dual canonical base 
of $U$ in $M\otimes N$.
The coefficients $c^{\mathbf{l,l'}}_{\mathbf{k,k'}}$ belong 
to $q^{-1}\mathbb{N}[q^{-1}]$ unless $c^{\mathbf{k,k'}}_{\mathbf{k,k'}}=1$.
\end{theorem}

\begin{proof}
The diagram $b\otimes b'$ is obtained by placing the diagrams 
$b$ and $b'$ in parallel. 
Note that the down arrows in $b$ is left to the up arrows in 
$b'$.
To obtain a diagram for a dual canonical basis of $U$, 
we use Lemma~\ref{moveuparrow} successively.
We can move 
an up arrow in $b'$ to the right of a down arrow in $b$.  
The coefficients in the right hand sides of Eqns.(\ref{moveuparroweq1}) and 
(\ref{moveuparroweq2})
are in $\mathbb{N}[q^{-1}]$. 
The coefficient of the arc and the dashed arc is one, that is, 
not in $q^{-1}\mathbb{N}[q^{-1}]$.
Thus, we have $c^{\mathbf{k,k'}}_{\mathbf{k,k'}}=1$ and 
$c^{\mathbf{l,l'}}_{\mathbf{k,k'}}\in q^{-1}\mathbb{N}[q^{-1}]$.
\end{proof}

\subsection{\texorpdfstring{Action of $Y$ on standard bases and dual canonical bases}
{Action of Y on standard bases and dual canonical bases}}
We consider the action of $Y$ on a standard basis 
$v^{\kappa}:=v^{\kappa_1}\otimes\cdots\otimes v^{\kappa_{N}}$ where 
$\kappa_{i}=\pm1$.
For each $1\le i\le N$, we define 
$Y_{(i)}(v^{\kappa}):=v^{\kappa_1}\otimes\cdots\otimes v^{\kappa_{i-1}}
\otimes v^{-\kappa_{i}}\otimes v^{\kappa_{i+1}}\otimes\cdots\otimes v^{\kappa_{N}}$.
Set $d_{i}:=\sum_{j=1}^{i}\kappa_{j}$. 
The action of $Y$ is defined by 
\begin{eqnarray}
\label{actionYonSB}
Y(v^{\kappa}):=\sum_{i=1}^{N}q^{d_{i-1}}Y_{(i)}(v^{\kappa})
+q^{d_{N}}v^{\kappa}.
\end{eqnarray}

\begin{prop}
The definition (\ref{actionYonSB}) provides the action of $Y$ 
on a standard basis.
\end{prop}
\begin{proof}
We prove Proposition by induction on $N$. 
When $N=1$, Proposition holds true by a straightforward calculation.
We assume that Proposition is true up to some $N\ge1$.
A standard basis $v^{\kappa}$ is written as 
$v^{\kappa}=v^{\kappa_1}\otimes v'$ where $v'$ is a standard 
basis of length $N-1$.
Let $d'_{i}:=\sum_{j=2}^{i}\kappa_j$.
From the induction assumption, we have 
\begin{eqnarray*}
Y(v')=\sum_{i=2}^{N}q^{d'_{i-1}}Y_{(i-1)}(v')+q^{d'_{N}}v'.
\end{eqnarray*}
From Eqns.(\ref{coproductY}) and (\ref{actionYonSB}), we have 
\begin{eqnarray*}
Y(v^{\kappa})&=&q^{\kappa_1}v^{\kappa_1}\otimes Y(v')
+v^{-\kappa_1}\otimes v' \\
&=&\sum_{i=2}^{N}q^{\kappa_1+d'_{i-1}}v^{\kappa_1}\otimes Y_{(i-1)}(v')
+q^{\kappa_1+d'_{N}}v \\
&=&\sum_{i=1}^{N}q^{d_{i-1}}Y_{(i)}(v)+q^{d_{N}}v,
\end{eqnarray*}
where we have used $v^{\kappa_1}\otimes Y_{(i-1)}(v')=Y_{(i)}(v)$.
\end{proof}

Recall that the diagram $D$ for a dual canonical basis of $U$ 
consists of arcs, dashed arcs, up arrows, down arrow 
with a star and at most one unpaired down arrow.
We enumerate all the up arrows from left to right 
by $1,2,\ldots,n_{\uparrow}$ where $n_{\uparrow}$ is the number 
of up arrows of $D$.

If $n_{\uparrow}=0$, we have three cases for a diagram $D$:
\begin{enumerate}
\item $D$ has no down arrow with a star. 
Define the action of $Y$ by 
\begin{eqnarray*}
Y(D):=D.
\end{eqnarray*}
\item $D$ has a down arrow with a star and an unpaired down arrow.
Let $D'$ be a diagram obtained from $D$ by changing the unpaired 
down arrow of $D$ to an up arrow. Define the action of $Y$ by 
\begin{eqnarray*}
Y(D):=D'.
\end{eqnarray*}
\item $D$ has a down arrow with a star but no unpaired down arrow.
Define the action of $Y$ by 
\begin{eqnarray*}
Y(D):=0.
\end{eqnarray*}
\end{enumerate}

Suppose $n_{\uparrow}>0$.  
For each $i$, $1\le i<n_{\uparrow}$, we denote by $Y_{(i)}(D)$ a  
diagram obtained from $D$ by connecting the $i$-th and $(i+1)$-th 
up arrows of $D$ via an unoriented arc.

We denote by $D'$ a diagram 
obtained from $D$ by changing the $n_{\uparrow}$-th up arrow to 
a down arrow. 
We construct $Y_{(n_{\uparrow})}(D)$ from $D'$ as follows.
If $D$ has an unpaired down arrow $d$, $Y_{(n_{\uparrow})}(D)$ is obtained 
from $D'$ by connecting the reversed down arrow of $D'$ and $d$ via a dashed arc.
If $D$ has no unpaired down arrow, then $Y_{(n_{\uparrow})}(D)=D'$.

If there is no down arrow with a star in $D$, we define 
$Y_{(n_{\uparrow}+1)}(D)=D$. 
If there is an unpaired down arrow in $D$, we denote by 
$Y_{(n_{\uparrow}+1)}(D)$ a diagram obtained from $D$ by changing
the unpaired down arrow to an up arrow. 
Define $Y_{(n_{\uparrow}+1)}(D)=0$ for otherwise.

Define the action of $Y$ by 
\begin{eqnarray}
\label{actionYonD}
Y(D):=\sum_{1\le i\le n_{\uparrow}+1}[i]Y_{(i)}(D).
\end{eqnarray} 

\begin{example}
Let $\mathbf{k}=(1,0,-2)\in I_{(3,4,4)}$.
The diagram $D$  is 
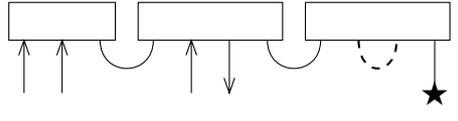
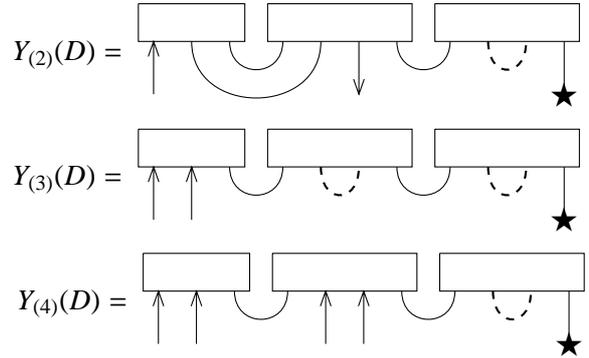
\begin{eqnarray*}
\begin{tikzpicture}
\draw(0,0)--(1.4,0)--(1.4,-0.5)--(0,-0.5)--(0,0);
\draw(0.2,-0.5)--(0.2,-1.2)(0.13,-0.7)--(0.2,-0.5)--(0.27,-0.7);
\draw(0.7,-0.5)--(0.7,-1.2)(0.63,-0.7)--(0.7,-0.5)--(0.77,-0.7);
\draw(1.7,0)--(3.6,0)--(3.6,-0.5)--(1.7,-0.5)--(1.7,0);
\draw(2.4,-0.5)--(2.4,-1.2)(2.33,-0.7)--(2.4,-0.5)--(2.47,-0.7);
\draw(2.9,-0.5)--(2.9,-1.2)(2.83,-1)--(2.9,-1.2)--(2.97,-1);
\draw(3.9,0)--(5.8,0)--(5.8,-0.5)--(3.9,-0.5)--(3.9,0);
\draw(5.6,-0.5)--(5.6,-1.2)node{$\bigstar$};
\draw(1.2,-0.5)..controls(1.2,-1)and(1.9,-1)..(1.9,-0.5);
\draw(3.4,-0.5)..controls(3.4,-1)and(4.1,-1)..(4.1,-0.5);
\draw[thick,dashed](4.6,-0.5)..controls(4.6,-1)and(5.1,-1)..(5.1,-0.5);
\end{tikzpicture}.
\end{eqnarray*}
Then 
\begin{eqnarray*}
Y_{(2)}(D)=
\raisebox{-0.5\totalheight}{
\begin{tikzpicture}
\draw(0,0)--(1.4,0)--(1.4,-0.5)--(0,-0.5)--(0,0);
\draw(0.2,-0.5)--(0.2,-1.2)(0.13,-0.7)--(0.2,-0.5)--(0.27,-0.7);
\draw(1.7,0)--(3.6,0)--(3.6,-0.5)--(1.7,-0.5)--(1.7,0);
\draw(2.9,-0.5)--(2.9,-1.2)(2.83,-1)--(2.9,-1.2)--(2.97,-1);
\draw(3.9,0)--(5.8,0)--(5.8,-0.5)--(3.9,-0.5)--(3.9,0);
\draw(5.6,-0.5)--(5.6,-1.2)node{$\bigstar$};
\draw(1.2,-0.5)..controls(1.2,-1)and(1.9,-1)..(1.9,-0.5);
\draw(3.4,-0.5)..controls(3.4,-1)and(4.1,-1)..(4.1,-0.5);
\draw[thick,dashed](4.6,-0.5)..controls(4.6,-1)and(5.1,-1)..(5.1,-0.5);
\draw(0.7,-0.5)..controls(0.7,-1.5)and(2.4,-1.5)..(2.4,-0.5);
\end{tikzpicture}}, \\
Y_{(3)}(D)=
\raisebox{-0.5\totalheight}{
\begin{tikzpicture}
\draw(0,0)--(1.4,0)--(1.4,-0.5)--(0,-0.5)--(0,0);
\draw(0.2,-0.5)--(0.2,-1.2)(0.13,-0.7)--(0.2,-0.5)--(0.27,-0.7);
\draw(0.7,-0.5)--(0.7,-1.2)(0.63,-0.7)--(0.7,-0.5)--(0.77,-0.7);
\draw(1.7,0)--(3.6,0)--(3.6,-0.5)--(1.7,-0.5)--(1.7,0);
\draw(3.9,0)--(5.8,0)--(5.8,-0.5)--(3.9,-0.5)--(3.9,0);
\draw(5.6,-0.5)--(5.6,-1.2)node{$\bigstar$};
\draw(1.2,-0.5)..controls(1.2,-1)and(1.9,-1)..(1.9,-0.5);
\draw(3.4,-0.5)..controls(3.4,-1)and(4.1,-1)..(4.1,-0.5);
\draw[thick,dashed](4.6,-0.5)..controls(4.6,-1)and(5.1,-1)..(5.1,-0.5);
\draw[thick,dashed](2.4,-0.5)..controls(2.4,-1)and(2.9,-1)..(2.9,-0.5);
\end{tikzpicture}}, \\
Y_{(4)}(D)=
\raisebox{-0.5\totalheight}{
\begin{tikzpicture}
\draw(0,0)--(1.4,0)--(1.4,-0.5)--(0,-0.5)--(0,0);
\draw(0.2,-0.5)--(0.2,-1.2)(0.13,-0.7)--(0.2,-0.5)--(0.27,-0.7);
\draw(0.7,-0.5)--(0.7,-1.2)(0.63,-0.7)--(0.7,-0.5)--(0.77,-0.7);
\draw(1.7,0)--(3.6,0)--(3.6,-0.5)--(1.7,-0.5)--(1.7,0);
\draw(2.4,-0.5)--(2.4,-1.2)(2.33,-0.7)--(2.4,-0.5)--(2.47,-0.7);
\draw(2.9,-0.5)--(2.9,-1.2)(2.83,-0.7)--(2.9,-0.5)--(2.97,-0.7);
\draw(3.9,0)--(5.8,0)--(5.8,-0.5)--(3.9,-0.5)--(3.9,0);
\draw(5.6,-0.5)--(5.6,-1.2)node{$\bigstar$};
\draw(1.2,-0.5)..controls(1.2,-1)and(1.9,-1)..(1.9,-0.5);
\draw(3.4,-0.5)..controls(3.4,-1)and(4.1,-1)..(4.1,-0.5);
\draw[thick,dashed](4.6,-0.5)..controls(4.6,-1)and(5.1,-1)..(5.1,-0.5);
\end{tikzpicture}}.
\end{eqnarray*}
We have $Y_{(1)}(D)=0$. 
Therefore,
\begin{eqnarray*}
Y(D)=[2]Y_{(2)}(D)+[3]Y_{(3)}(D)+[4]Y_{(4)}(D).
\end{eqnarray*}

\end{example}

\begin{theorem}
\label{thm-actionT}
The definition (\ref{actionYonD}) provides the action of $Y$ on a 
canonical basis.
\end{theorem}
\begin{proof}
It is enough to show the case of $\mathbf{m}=(1,\ldots,1)$ 
since the actions of generators of $U_q(\mathfrak{sl}_2)$  
and that of projector commute with each other.
Hereafter, we set $\mathbf{m}=(1,\ldots,1)$.
We prove Theorem by induction.
In the case of $L=1,\ldots,4$, we can verify Theorem by a direct computation.
Suppose that Theorem is true up to some $L\ge4$.

Suppose the leftmost arrow is an up arrow. 
Let $D$ be a diagram $\uparrow\pi_0$ where $\pi_0$ is a diagram
of length $L-1$. 
\begin{eqnarray}
\label{actionT0}
\begin{aligned}
Y(D)&=
(K\otimes t+q^{-1}KE\otimes1+F\otimes1)\uparrow\pi_0 \\
&=q\uparrow Y(\pi_0)+\downarrow\pi_0
\end{aligned}
\end{eqnarray}
We have the following six cases:
\begin{enumerate}[(A)]
\item $\pi_0$ has no up arrows.
\begin{enumerate}[({A}1)] 
\item $\pi_0$ has no down arrow with a star.
The diagram $\pi_0$ consists of arcs. 
\item $\pi_0$ has a down arrow with a star and an unpaired arrow.
The diagram $\pi_0$ is written as $\pi_0=\pi_1\downarrow\pi_2$ 
where $\pi_1$ consists of 
arcs and $\pi_2$ consists of arcs, dashed arcs and down arrow with a star.
\item $\pi_0$ has a down arrow with a star but no unpaired arrow.
The diagram $\pi_0$ consists of arcs, dashed arcs and a down arrow with 
a star.
\end{enumerate}
\item $\pi_0$ has up arrows. 
\begin{enumerate}[(B1)]
\item $\pi_0$ has no down arrow with a star.
The diagram $\pi_0$ is written as 
$\pi_0=\pi_1\uparrow\pi_2\uparrow\ldots\uparrow\pi_{l}$  where 
$\pi_i,1\le i\le l$, consists of arcs.
\item $\pi_0$ has a down arrow with a star and an unpaired arrow.
The diagram $\pi_0$ is written as 
$\pi_0=\pi_1\uparrow\pi_2\uparrow\ldots\uparrow\pi_{l}\downarrow\pi_{l+1}$
where $\pi_i, 1\le i\le l$, consists of arcs and $\pi_{l+1}$ consists of 
arcs, dashed arcs and a down arrow with a star.
\item $\pi_0$ has a down arrow with a star but no unpaired arrow.
The diagram $\pi_0$ is written as 
$\pi_0=\pi_1\uparrow\pi_2\uparrow\ldots\uparrow\pi_{l}$  where 
$\pi_i,1\le i\le l-1$, consists of arcs and $\pi_l$ consists of 
arcs, dashed arcs and a down arrow with a star.
\end{enumerate}
\end{enumerate}

We consider (B2) case since all other cases can be similarly proven.
From the assumption, we have 
\begin{eqnarray}
\label{actionT1}
\begin{aligned}
Y(\pi_0)
&=
\sum_{2\le i\le l-1}[i-1]\cdot 
\pi_1\uparrow\ldots\uparrow\pi_{i-1}
\raisebox{-0.6\totalheight}{
\begin{tikzpicture}
\draw(0,-0.1)..controls(0,-0.8)and(0.8,-0.8)..(0.8,-0.1);
\draw(0.4,-0.25)node{$\pi_{i}$};
\end{tikzpicture}}
\ \pi_{i+1}\uparrow\ldots\uparrow\pi_{l}\downarrow\pi_{l+1} \\
&\quad+
[l-1]\cdot\uparrow\pi_1\uparrow\ldots\uparrow\pi_{l-1}
\raisebox{-0.6\totalheight}{
\begin{tikzpicture}
\draw[thick,dashed](0,-0.1)..controls(0,-0.8)and(0.8,-0.8)..(0.8,-0.1);
\draw(0.4,-0.25)node{$\pi_{l}$};
\end{tikzpicture}}
\ \pi_{l+1}
+[l]\cdot\pi_{1}\uparrow\ldots\uparrow\pi_{l+1}.
\end{aligned}
\end{eqnarray}
We also have 
\begin{eqnarray}
\label{actionT2}
\begin{aligned}
\downarrow\pi_1\uparrow\ldots\uparrow\pi_{l}\downarrow\pi_{l+1}
&=
\sum_{1\le i\le l-1}
q^{-(i-1)}\cdot
\uparrow\pi_1\ldots
\raisebox{-0.6\totalheight}{
\begin{tikzpicture}
\draw(0,-0.1)..controls(0,-0.8)and(0.8,-0.8)..(0.8,-0.1);
\draw(0.4,-0.25)node{$\pi_{i}$};
\end{tikzpicture}}
\ldots\pi_l\downarrow\pi_{l+1} \\
&\quad+
q^{-(l-1)}\cdot
\uparrow\pi_1\ldots\pi_{l-1}
\raisebox{-0.6\totalheight}{
\begin{tikzpicture}
\draw[thick,dashed](0,-0.1)..controls(0,-0.8)and(0.8,-0.8)..(0.8,-0.1);
\draw(0.4,-0.25)node{$\pi_{l}$};
\end{tikzpicture}}
\ \pi_{l+1}
+q^{-l}\cdot \uparrow\pi_1\uparrow\ldots\uparrow\pi_{l+1}.
\end{aligned}
\end{eqnarray}
Substituting Eqns.(\ref{actionT1}) and (\ref{actionT2}) into 
Eqn.(\ref{actionT0}), we obtain
\begin{eqnarray*}
\begin{aligned}
Y(D)
&=
\sum_{1\le i\le l-1}
[i]\cdot
\uparrow\pi_1\uparrow\ldots
\raisebox{-0.6\totalheight}{
\begin{tikzpicture}
\draw(0,-0.1)..controls(0,-0.8)and(0.8,-0.8)..(0.8,-0.1);
\draw(0.4,-0.25)node{$\pi_{i}$};
\end{tikzpicture}}
\ldots\pi_{l}\downarrow\pi_{l+1} \\
&+
[l]\cdot
\uparrow\pi_1\ldots\pi_{l-1}
\raisebox{-0.6\totalheight}{
\begin{tikzpicture}
\draw[thick,dashed](0,-0.1)..controls(0,-0.8)and(0.8,-0.8)..(0.8,-0.1);
\draw(0.4,-0.25)node{$\pi_{l}$};
\end{tikzpicture}}
\ \pi_{l+1}
+
[l+1]\cdot
\uparrow\pi_{1}\uparrow\ldots\uparrow\pi_{l+1},
\end{aligned}
\end{eqnarray*}
where we have used $q[i]+q^{-i}=[i+1]$.
This is the desired expression for $Y(D)$.

Suppose the leftmost arrow $a$ is a down arrow.
Then $D$ has the following four possibilities:
\begin{enumerate}[(A)]
\item $a$ is connected with an up arrow by an arc. 
The diagram $D$ is written as 
$
\raisebox{-0.6\totalheight}{
\begin{tikzpicture}
\draw(0,-0.1)..controls(0,-0.8)and(0.8,-0.8)..(0.8,-0.1);
\draw(0.4,-0.25)node{$\pi_{0}$};
\end{tikzpicture}}
\ \pi_1
$
where $\pi_0$ consists of arcs and $\pi_1$ consists of arcs, 
dashed arcs, a down arrow with a star and an unpaired down arrow.

\item $a$ is a down arrow with a star. 
$D$ is written as 
$
\raisebox{-0.5\totalheight}{
\begin{tikzpicture}
\draw(0,0)--(0,-0.4)node{$\bigstar$};
\end{tikzpicture}}
\pi_0
$ 
where $\pi_0$ consists of arcs.

\item $a$ is connected with an down arrow by a dashed arc.
The diagram $D$ is written as 
$
\raisebox{-0.6\totalheight}{
\begin{tikzpicture}
\draw[thick,dashed](0,-0.1)..controls(0,-0.8)and(0.8,-0.8)..(0.8,-0.1);
\draw(0.4,-0.25)node{$\pi_{0}$};
\end{tikzpicture}}
\ \pi_1
$
where $\pi_0$ consists of arcs and $\pi_1$ consists of arcs, dashed arcs and 
a down arrow with a star.

\item $a$ is an unpaired down arrow. 
The diagram $D$ is depicted as $\downarrow\pi_0$ where $\pi_0$ consists of 
arcs, dashed arc and a down arrow with a star.
\end{enumerate} 
We consider the case (C) since other cases are similarly proven.
We have
\begin{eqnarray*}
Y(D)&=&Y(\downarrow\pi_0\downarrow\pi_1-q^{-1}\uparrow\pi_0\uparrow\pi_1) \\
&=&
q^{-1}\downarrow Y(\pi_0\downarrow\pi_1)+\uparrow\pi_0\downarrow\pi_1
-\uparrow Y(\pi_0\downarrow\pi_1)-q^{-1}\downarrow\pi_0\downarrow\pi_1 \\
&=&0.
\end{eqnarray*}
Here we have used $Y(\pi_0\downarrow\pi_1)=\pi_0\uparrow\pi_1$ since 
the down arrow in $\pi_0\downarrow\pi_1$ is an unpaired down arrow in
the diagram $\pi_0\downarrow\pi_1$.
This is the desired expression for $Y(D)$.
\end{proof}
\begin{cor}
\label{cor-Y1}
All coefficients of the action of $Y$ on the dual canonical basis
in a tensor product $V_{m_1}\otimes\ldots\otimes V_{m_n}$ belong
to $\mathbb{N}[q,q^{-1}]$. Especially, they are quantum integers.
\end{cor}

\subsection{\texorpdfstring{Eigensystem of $Y$}{Eigensystem of Y}}
Let $\mathbf{k}\in I_{\mathbf{m}}$ and $D$ be a diagram for  
a dual canonical base 
$v^{k_1}\varspadesuit\ldots\varspadesuit v^{k_{n}}$.
The integer $n_{\uparrow}$ is the number of up arrows in $D$.
We define an integer $N_{\mathbf{k}}$ as follows:
\begin{enumerate}
\item If $D$ does not have a down arrow with a star, 
$N_{\mathbf{k}}=n_{\uparrow}+1$. 
\item If $D$ has a down arrow with a star but no unpaired down arrow,
$N_{\mathbf{k}}=n_{\uparrow}$.
\item If $D$ has a down arrow with a star and an unpaired down arrow,
$N_{\mathbf{k}}=-(n_{\uparrow}+1)$.
\end{enumerate}
Note that $|N_{\mathbf{k}}|$ is the maximum integer which appears 
in the expansion of $Y(D)$.
For an integer $j\in\mathbb{Z}$, we define 
\begin{eqnarray*}
M_{j}:=\#\{N_{\mathbf{k}}| N_{\mathbf{k}}=j \text{ and } \mathbf{k}\in I_{\mathbf{m}} \}.
\end{eqnarray*}
For example, consider $I_{(2,2)}$. Then, $N_{(2,-2)}=-3$ and $N_{(0,-2)}=1$. 
We have $M_5=1, M_3=2, M_1=3, M_{-1}=2$ and $M_{-3}=1$. 

We consider the matrix representation of $Y$ in $V_{m_1}\otimes\ldots\otimes V_{m_n}$ 
with respect to dual canonical bases.
\begin{theorem}
\label{thm-eigenvalues}
The generator $Y$ has an eigenvalue $[j], j\in\mathbb{Z}$, with the multiplicity $M_{j}$.
\end{theorem}

Before proceeding to the proof of Theorem~\ref{thm-eigenvalues}, 
we introduce notations and lemmas.

We define the lexicographic order for $\mathbf{k}\in I_{\mathbf{m}}$:  
$\mathbf{k'}<_{\mathrm{lex}}\mathbf{k}$ if and only if 
there exists $i$ such that 
$k'_{j}=k_{j}$ for $1\le j\le i$ and 
$k'_{i+1}<k_{i+1}$.

Let $\Gamma$ be a graph whose vertices are indexed by 
$\mathbf{k}\in I_{\mathbf{m}}$.
We connect vertices $\mathbf{k}$ and $\mathbf{k'}$ by 
an arrow from $\mathbf{k}$ to $\mathbf{k'}$ if 
$v^{k'_1}\varspadesuit\ldots\varspadesuit v^{k'_n}$ appears 
in the expansion of 
$Y(v^{k_1}\varspadesuit\ldots\varspadesuit v^{k_n})$.
If there exists an arrow from $\mathbf{k}$ to $\mathbf{k'}$, 
then we denote it by $\mathbf{k}\rightarrow\mathbf{k'}$.
Suppose $D$ and $D'$ are the diagrams associated with $\mathbf{k}$
and $\mathbf{k'}$.
We write $D\rightarrow D'$ if $\mathbf{k}\rightarrow\mathbf{k'}$.

It is obvious that
\begin{lemma}
The vertex $(m_1,\ldots,m_n)\in I_{\mathbf{m}}$ has a unique incoming 
arrow from itself.
A vertex $\mathbf{k}$ with $\mathbf{k}\neq(m_1,\ldots,m_n)$ 
has at least one incoming arrows.
\end{lemma}

\begin{lemma}
\label{lemma-Gamma}
In $\Gamma$, suppose $\mathbf{k'}\rightarrow\mathbf{k}$
with $\mathbf{k'}<_{lex}\mathbf{k}$ and 
$\mathbf{k,k'}\neq(m_1,\ldots,m_n)$.
Then we have
\begin{enumerate}
\item The number of such $\mathbf{k}'$ is at most one for 
each vertex $\mathbf{k}$.
\item The vertex $\mathbf{k'}$ has a incoming arrow from 
the vertex $\mathbf{k}$, that is, $\mathbf{k}\rightarrow\mathbf{k'}$.
\item The vertices $\mathbf{k,k'}$ do not have an arrow from itself.
\item There is no $\mathbf{k''}$ such that 
$\mathbf{k''}\rightarrow\mathbf{k'}$ and $\mathbf{k''}<_{\mathrm{lex}}\mathbf{k'}$
\end{enumerate}
\end{lemma}
\begin{proof}
Recall the explicit action of $Y$ in Theorem~\ref{thm-actionT}. 
The relation 
$\mathbf{k'}\rightarrow\mathbf{k}$ with $\mathbf{k'}<_{\mathrm{lex}}\mathbf{k}$ 
means that the diagram $D$ of $v^{k_1}\varspadesuit\ldots\varspadesuit v^{k_n}$ 
has a down arrow with a star but not an unpaired down arrow. 
Consider the diagram $D'$ obtained from $D$ by 
changing the rightmost up arrow to an unpaired down arrow.
We have $\mathbf{k'}<_{\mathrm{lex}}\mathbf{k}$ and $D$ appears 
in the expansion of $Y(D')$.
Similarly, let $D''$ be a diagram obtained from $D$ by changing 
the leftmost down arrow (which forms a dashed arc or a down arrow
with a star) to an up arrow, or by changing an outer arc to 
two up arrows.
In these two cases, we have $\mathbf{k}<_{\mathrm{lex}}\mathbf{k''}$. 
Thus (1) follows.
The statements (2) is obvious since $D'$ appears in 
the expansion $Y(D)$.
The statement (3) follows from the fact that $D$ (resp. $D'$) 
does not appear in the expansion of $Y(D)$ (resp. $Y(D')$).
We can show the statement (4) by a similar argument on $D'$.
\end{proof}

\begin{proof}[Proof of Theorem~\ref{thm-eigenvalues}]
We will construct an eigenvector characterized by $\mathbf{k}$ with
the eigenvalue $N_{\mathbf{k}}$. 

Let $a_{\mathbf{k}}$ be an indeterminate and take a linear combination 
of the dual canonical basis,
\begin{eqnarray*}
\chi
:=
\sum_{\mathbf{k}}
a_{\mathbf{k}}v^{k_1}\varspadesuit\cdots\varspadesuit v^{k_n}.
\end{eqnarray*}
From Corollary~\ref{cor-Y1}, the action of $Y$ on a dual canonical
basis is written as 
\begin{eqnarray*}
Y(v^{k_1}\varspadesuit\cdots\varspadesuit v^{k_n})=
\sum_{\mathbf{k'}}[n_{\mathbf{k',k}}]
v^{k'_1}\varspadesuit\cdots\varspadesuit v^{k'_n}
\end{eqnarray*}
with a non-negative integer $n_{\mathbf{k',k}}$. 
We want to solve the eigenvalue problem $Y\chi=y\chi$ 
(with an eigenvalue $y$).
The eigenvalue problem is equivalent to 
\begin{eqnarray}
\label{EP1}
\sum_{\mathbf{k'}:\mathbf{k'}\rightarrow\mathbf{k}}
[n_{\mathbf{k},\mathbf{k'}}]a_{\mathbf{k'}}=ya_{\mathbf{k}}.
\end{eqnarray}

Let $D$ be the diagram for $v^{k_1}\varspadesuit\cdots\varspadesuit v^{k_n}$
and $\Gamma'$ be a graph obtained from $\Gamma$ by deleting 
an outgoing arrow from a vertex $\mathbf{k'}$ if $a_{\mathbf{k'}}=0$.
Suppose $a_{\mathbf{k}}\neq0$. 
We have three cases for $D$: (A) $D$ has no down arrow with a star,  
(B) $D$ has a down arrow with a star but not an unpaired down arrow, 
and (C) $D$ has a down arrow with a star and an unpaired down arrow.

\paragraph{\bf Case A} 
We set $a_{\mathbf{k'}}=0$ if $|N_{\mathbf{k'}}|\ge N_{\mathbf{k}}$.
In $\Gamma'$, there exists no incoming arrow on a vertex $\mathbf{k}$ 
except from itself.
The graph $\Gamma'$ does not have an arrow from $\mathbf{k'}$ with 
$\mathbf{k}<_{\mathrm{lex}}\mathbf{k'}$.
We solve the eigenvalue problem by induction in the lexicographic
order. 

The $\mathbf{k}$-component of Eqn.(\ref{EP1}) is equal to 
\begin{eqnarray*}
[N_{\mathbf{k}}]a_{\mathbf{k}}=ya_{\mathbf{k}}.
\end{eqnarray*}
The solution is $y=[N_{\mathbf{k}}]$ since $a_{\mathbf{k}}\neq0$.

Suppose that we have a partial solution 
$\{a_{\mathbf{l'}}| \mathbf{l}<_{\mathrm{lex}}\mathbf{l'}\le_{\mathrm{lex}}\mathbf{k}\}$
for some $\mathbf{l}$.
We will show that $a_{\mathbf{l}}$ is uniquely obtained from 
the partial solution.
We have two cases from Lemma~\ref{lemma-Gamma}. 
\begin{enumerate}
\item The vertex $\mathbf{l}$ has no incoming arrow from 
$\mathbf{l'}$ with $\mathbf{l'}<_{\mathrm{lex}}\mathbf{l}$. 
Equation (\ref{EP1}) is equal to 
\begin{eqnarray*}
\sum_{\genfrac{}{}{0pt}{}{\mathbf{l'}\neq\mathbf{l}}{\mathbf{l'}\rightarrow\mathbf{l}}}
[n_{\mathbf{l},\mathbf{l'}}]a_{\mathbf{l'}}
=
(y-[n_{\mathbf{l},\mathbf{l}}])a_{\mathbf{l}}.
\end{eqnarray*}
Since the integer $n_{\mathbf{l},\mathbf{l}}$ is zero or a positive integer 
less than $N_{\mathbf{k}}$, $a_{\mathbf{l}}$ is uniquely written in terms of the partial 
solution $\{a_{\mathbf{l'}}| \mathbf{l}<_{\mathrm{lex}}\mathbf{l'}\}$.

\item The vertex $\mathbf{l}$ has a incoming arrow from 
$\mathbf{l'}$ with $\mathbf{l'}<_{\mathrm{lex}}\mathbf{l}$. 
Equation (\ref{EP1}) is equal to 
\begin{eqnarray*}
\sum_{\genfrac{}{}{0pt}{}{\mathbf{l''}\neq\mathbf{l,l'}}{\mathbf{l''}\rightarrow\mathbf{l}}}
[n_{\mathbf{l},\mathbf{l''}}]a_{\mathbf{l''}}
&=&
ya_{\mathbf{l}}-[n_{\mathbf{l},\mathbf{l'}}]a_{\mathbf{l'}}, \\
\sum_{\genfrac{}{}{0pt}{}{\mathbf{l''}\neq\mathbf{l,l'}}{\mathbf{l''}\rightarrow\mathbf{l'}}}
[n_{\mathbf{l},\mathbf{l''}}]a_{\mathbf{l''}}
&=&
ya_{\mathbf{l'}}-[n_{\mathbf{l'},\mathbf{l}}]a_{\mathbf{l}}.
\end{eqnarray*}
where $\mathbf{l}<_{\mathrm{lex}}\mathbf{l''}$. 
Since $n_{\mathbf{k},\mathbf{k'}},n_{\mathbf{k'},\mathbf{k}}$ are less than $N_{\mathbf{k}}$,
the rank of the equations is two. We can solve these two equations simultaneously 
with respect to $a_{\mathbf{l}}$ and $a_{\mathbf{l'}}$.
\end{enumerate} 
In both cases, $a_{\mathbf{l}}$ is expressed in terms of the partial solution.
An eigenvector with an eigenvalue $N_{\mathbf{k}}$ is constructed.

\paragraph{\bf Case B}
Let $D'$ be the diagram for $v^{k'_1}\varspadesuit\cdots\varspadesuit v^{k'_n}$ 
obtained from $D$ by changing the rightmost up arrow to an 
unpaired down arrow.
We set $a_{\mathbf{k'}}\neq0$ and 
$a_{\mathbf{k''}}=0$  if $|N_{\mathbf{k''}}|\ge N_{\mathbf{k}}$.
In $\Gamma'$, the vertices $\mathbf{k,k'}$ have only one incoming arrow, 
namely, $\mathbf{k}\leftrightarrow\mathbf{k'}$.
The $\mathbf{k}$- and $\mathbf{k'}$-components in Eqn.(\ref{EP1}) are equal to
\begin{eqnarray*}
[|N_{\mathbf{k'}}|]a_{\mathbf{k'}}=ya_{\mathbf{k}}, \\
{[N_{\mathbf{k}}]}a_{\mathbf{k}}=ya_{\mathbf{k'}}.
\end{eqnarray*}
Note $N_{\mathbf{k}}>0$. 
Since $N_{\mathbf{k'}}=-N_{\mathbf{k}}$, we have the solution 
$y=\pm[N_{\mathbf{k}}]$.
We take $y=[N_{\mathbf{k}}]$ since $y=-[N_{\mathbf{k}}]=[N_{\mathbf{k'}}]$ 
will be considered in Case C.
By a similar argument to Case A, we can construct an eigenvector
with an eigenvalue $[N_{\mathbf{k}}]$.

\paragraph{\bf Case C}
Let $D'$ be the diagram for $v^{k'_1}\varspadesuit\cdots\varspadesuit v^{k'_n}$ 
such that $D'$ is obtained from $D$ by changing the unpaired down arrow to 
an up arrow.
We set $a_{\mathbf{k'}}\neq0$ and $a_{\mathbf{k''}}=0$  if 
$|N_{\mathbf{k''}}|\ge N_{\mathbf{k}}$.
Note that $N_{\mathbf{k}}<0$.
By a similar argument to Case (B), we can construct an eigenvector with
an eigenvalue $N_{\mathbf{k}}$.

In the cases B and C, an eigenvector is characterized
by $\mathbf{k}$ and $\mathbf{k'}$. However, they are distinguished by 
the sign of an eigenvalue.
Therefore, in all cases, an eigenvector is characterized by $\mathbf{k}$
by construction. 
This completes the proof of Theorem.
\end{proof}

Set $L=\sum_{i=1}^{n}m_i$ for $\mathbf{m}=(m_1,\ldots,m_n)$. 
\begin{cor}
\label{cor-EP}
The generator $Y$ has the eigenvector with the eigenvalue 
$[L+1]$ and the multiplicity is one. 
\end{cor}
\begin{proof}
Choose $\mathbf{k}=(m_1,\ldots,m_n)$. We have $N_{\mathbf{k}}=L+1$.
The number of up arrows in the diagram for $\mathbf{k'}\neq\mathbf{k}$ 
is less than $L$. Thus, we have $M_{L+1}=1$.
\end{proof}

\begin{cor}
The eigenvalue $[L+1]$ is the largest at $q=1$.
\end{cor}
Let $\Psi:=\sum_{\mathbf{k}\in I_{\mathbf{m}}}
\Psi_{\mathbf{k}}v^{k_1}\varspadesuit\ldots\varspadesuit v^{k_n}$ 
be the eigenvector of $Y$ with the eigenvalue $[L+1]$.
\begin{lemma}
The $(-m_1,\ldots,-m_n)$-component of $\Psi$ is non-zero.
\end{lemma}
\begin{proof}
We make use of the notation in the proof of Theorem~\ref{thm-eigenvalues}.
There exists a sequence of $\mathbf{k}_i$ such that 
\begin{eqnarray*}
\mathbf{k}_0=(m_1,\ldots,m_n)\rightarrow \mathbf{k}_1\rightarrow \mathbf{k}_2
\rightarrow\ldots\rightarrow\mathbf{k}_{L}=(-m_1,\ldots,-m_n).
\end{eqnarray*}
Note that this sequence is unique. 
For some $\mathbf{k}_i$, we have $\mathbf{k}_{i}\leftarrow\mathbf{k}_{i+1}$.
For each $\mathbf{k}_i$, there is no $\mathbf{k'}$ such that 
$\mathbf{k'}\neq\mathbf{k}_{i-1},\mathbf{k}_{i+1}$ and $\mathbf{k'}\rightarrow\mathbf{k}_i$.
From Lemma~\ref{lemma-Gamma}, there exists no triplet $(i,i+1,i+2)$ such that 
$\mathbf{k}_{i}\leftarrow\mathbf{k}_{i+1}\leftarrow\mathbf{k}_{i+2}$.
There is no incoming arrow from itself except $\mathbf{k}_0$.

From Eqn.(\ref{EP1}), when 
$\mathbf{k}_{i-1}\rightarrow\mathbf{k}_{i}\rightarrow\mathbf{k}_{i+1}$, 
we have
\begin{eqnarray*}
[n_{\mathbf{k}_{i-1},\mathbf{k}_{i}}]a_{\mathbf{k}_{i-1}}
=
[L+1]a_{\mathbf{k}_{i}}.
\end{eqnarray*}
When 
$\mathbf{k}_{i-1}\rightarrow\mathbf{k}_{i}\leftrightarrow
\mathbf{k}_{i+1}\rightarrow\mathbf{k}_{i+2}$, 
we have 
\begin{eqnarray*}
[n_{\mathbf{k}_{i-1},\mathbf{k}_i}]a_{\mathbf{k}_{i-1}}
+[n_{\mathbf{k}_{i+1},\mathbf{k}_i}]a_{\mathbf{k}_{i+1}}
&=&
[L+1]a_{\mathbf{k}_{i}}, \\
\left[n_{\mathbf{k}_{i},\mathbf{k}_{i+1}}\right]a_{\mathbf{k}_{i}}
&=&
[L+1]a_{\mathbf{k}_{i+1}}.
\end{eqnarray*}
In both cases, $a_{\mathbf{k}_i}$ is proportional to $a_{\mathbf{k}_{i-1}}$.
Therefore, $a_{(-m_1,\ldots,-m_n)}\neq0$ if $a_{(m_1,\ldots,m_n)}\neq0$.
The unique eigenvector $\Psi$ has non-zero $a_{(m_1,\ldots,m_n)}$.
This completes the proof.
\end{proof}

\begin{defn}
\label{def-normalization}
We normalize $\Psi$ such that 
$\Psi_{(-m_1,\ldots,-m_n)}=1$.
\end{defn}

Let $D$ be a diagram and $E$ be a diagram obtained from $D$ by removing
the projectors.
The graph $\Gamma$ for $D$ is obtained from the graph $\Gamma'$ for $E$
by deleting some vertices.
The eigenvalue problem for $\Psi_D$ is the same as the 
one for $\Psi_{E}$.
Thus this implies 
\begin{lemma}
\label{lemma-psi-noproj}
We have $\Psi_{D}=\Psi_{E}$.
\end{lemma}

\subsection{\texorpdfstring{Eigenvector of $Y$}{Eigenvector of Y}}
\label{sec-EigensystemY}
We first consider the eigenfunction of $Y$ with the eigenvalue $[N+1]$
on standard bases.
Then, we will show an explicit expression of $\Psi_{\mathbf{k}}$ on
the Kazhdan--Lusztig bases.

Let $v^{\kappa}:=v^{\kappa_1}\otimes\cdots\otimes v^{\kappa_N}$ with 
$\kappa_i=\pm1$, and $J:=\{i| \kappa_i=1\}$. 
We define a vector indexed by a binary string $\kappa$ by 
\begin{eqnarray}
\Psi^{0}_{\kappa}:=\prod_{i\in J}q^{N+1-i},
\end{eqnarray}
and $\Psi^{0}:=\sum_{\kappa}\Psi^{0}_{\kappa}v^{\kappa}$.
\begin{prop}
We have $Y\Psi^{0}=[N+1]\Psi^{0}$.
\end{prop}
\begin{proof}
In general, a binary string $\kappa$ is written as 
\begin{eqnarray*}
\underbrace{+\ldots+}_{n_1}
\underbrace{-\ldots-}_{n'_1}
\underbrace{+\ldots+}_{n_2}
\ldots
\underbrace{+\ldots+}_{n_I}
\underbrace{-\ldots-}_{n'_I}.
\end{eqnarray*}
We set $N_{\uparrow}:=\sum_{i=1}^{I}n_i$,  
$N_{\downarrow}:=\sum_{i=1}^{I}n'_i$ and 
$J_{i}:=\{j\in\mathbb{N}| 1+\sum_{k=1}^{i-1}(n_k+n'_k)\le j\le 
n_i+\sum_{k=1}^{i-1}(n_k+n'_k)\}$.
We will show 
$\sum_{\kappa'}Y_{\kappa,\kappa'}\Psi^{0}_{\kappa'}
=[N+1]\Psi^{0}_{\kappa}$ where $Y_{\kappa,\kappa'}$ is a matrix
representation of $Y$.
From the action of $Y$ on $v^{\kappa}$, $v^{\kappa}$ appears 
in the expansion of $Y(v^{\kappa'})$ if and only if there exists 
$i$ such that $\kappa'_{i}=-\kappa_{i}$ and 
$\kappa_{j}=\kappa'_{j}$ for $j\neq i$, or $\kappa'=\kappa$ .

The binary string $\kappa'$ with $\kappa\neq\kappa'$ and  
$Y_{\kappa,\kappa'}\neq0$ can be obtained from $\kappa$ 
by changing $\kappa_i$ to $-\kappa_i$ for $1\le i\le N$.
We have two cases for $\kappa_i$: 1) $\kappa_i=+$ and 
2) $\kappa_i=-$.

\paragraph{Case 1}
When $j\in J_{i}$, $1\le i\le I$, we have 
$Y_{\kappa,\kappa'}=q^{j-1-2\sum_{k=1}^{i-1}n'_{k}}$ and 
$\Psi_{\kappa'}=q^{-(N+1-j)}\Psi_{\kappa}$. 
The sum of contributions of these $\kappa'$'s is given by
\begin{eqnarray*}
\Psi^{0}_{\kappa}\sum_{i=1}^{I}\sum_{j\in J_{i}}q^{-N-2-2j-2\sum_{k=1}^{i-1}n'_{k}}
&=&\sum_{i=1}^{N_{\uparrow}}q^{-N-2+2j}\Psi^{0}_{\kappa} \\
&=&q^{-N+N_{\uparrow}-1}[N_{\uparrow}]\Psi^{0}_{\kappa}
\end{eqnarray*}

\paragraph{Case 2}
Let $J'_{i}:=\{j\in\mathbb{N}|1+n_i+\sum_{k=1}^{i-1}(n_k+n'_k)\le j\le
\sum_{k=1}^{i}(n_k+n'_{k})\}$.
By a similar argument to Case 1, the sum of contributions is given by
\begin{eqnarray*}
\Psi^{0}_{\kappa}\sum_{i=1}^{I}\sum_{j\in J'_{i}}q^{N+2-2j+2\sum_{k=1}^{i}n_{k}} 
=q^{N-N_{\downarrow}+1}[N_{\downarrow}]\Psi^{0}_{\kappa}.
\end{eqnarray*}

Since $Y_{\kappa,\kappa}=q^{N_{\uparrow}-N_{\downarrow}}$, we have a 
contribution from $\kappa$ itself, which is 
$q^{N_{\uparrow}-N_{\downarrow}}\Psi^{0}_{\kappa}$. 
The sum of these three contributions gives $[N+1]\Psi^{0}_{\kappa}$. 
This completes the proof.
\end{proof}

Below, we will consider the eigenvector $\Psi$ on the 
Kazhdan--Lusztig bases.
From Lemma~\ref{lemma-psi-noproj}, it is enough to consider 
$\Psi_{\mathbf{k}}$ with $k_i\in\{1,-1\}$ for all $1\le i\le n$ 
instead of $\Psi_{\mathbf{k'}}$ with $\mathbf{k'}\in I_{\mathbf{m}}$.

Let $D$ be a diagram for a dual canonical basis, $N$ be the total number of 
all arrows and $N_1$ is the number of an unpaired down arrow 
($N_1$ is either $0$ or $1$).
Let $S$ be the set of all arcs of $D$.
We define the sets of arcs $S_R, S_M$ and $S_L$ as follows. 
If $D$ has a down arrow with a star, we define
\begin{eqnarray*}
S_R&:=&\{A\in S|\text{$A$ is right to the down arrow with a star}\}, \\
S_M&:=&
\begin{cases}
\{A\in S\setminus S_R|\text{$A$ is right to the unpaired down arrow}\}, & \text{for $N_1=1$}, \\
\{A\in S\setminus S_R|\text{$A$ is right to the leftmost down arrow forming a dashed arc}\} &
\text{for $N_1=0$}, 
\end{cases} \\
S_L&:=&
\begin{cases}
\{A\in S\setminus(S_R\cup S_M)|\text{$A$ is right to the rightmost up arrow}\}, & 
\text{if $D$ has an up arrow}, \\
\{A\in S\setminus(S_R\cup S_M)\}, & \text{if $D$ has no up arrows}
\end{cases}
\end{eqnarray*}
and otherwise $S_R, S_M$ and $S_L$ are defined as the empty set. 
An arc $A$ is said to be an {\it outer arc} if there exists neither an arc nor 
a dashed arc outside of $A$.
We denote by $S^{+}$ the set of outer arcs and 
define $S^{+}_R:=S_R\cap S^{+}, S^{+}_M:=S_M\cap S^{+}$ and $S^{+}_L:=S_L\cap S^{+}$.
We define $T$ as the set of all dashed arcs of $D$.

We define
\begin{eqnarray*}
N_2&:=&
\begin{cases}
|S|, & \text{$D$ does not have a down arrow with a star}, \\
|S|+|T|+1 & \text{$D$ has a down arrow with a star},
\end{cases} \\
N_3&:=&|S_M|+|T|,\\
N_4&:=&N-|S|+|S_R|+|S_M|+|S_L|+1, \\
N_5&:=&
\begin{cases}
[N_4]/[N_3] & \text{$D$ has a down arrow with a star and $N_1=0$}, \\
1 & \text{otherwise}.
\end{cases}
\end{eqnarray*}

We enumerate all arrows (including arrows forming an arc, a dashed arc and 
a down arrow with a star) from left to right.
Let $s$ be the integer assigned to the down arrow with a star.
If the $i$-th down arrow and the $j$-th ($i<j$) up arrow form an arc $A$, then 
$A$ is said to be size $(j-i+1)/2$.
We denote by $m_A$ the size of an arc $A$.
For an arc $A$, we define the following values:
\begin{eqnarray*}
&&d_{1,A}:=N-j, \qquad
d_{2,A}:=\frac{i-s+1}{2}. \\
&&N_6:=
\begin{cases}
\displaystyle
\prod_{A\in S_L^+\cup S^+_M}
\frac{[1+m_A+d_{1,A}]}{[1+2m_A+d_{1,A}]}, & N_1=0, \\
\displaystyle
\prod_{A\in S^+_M}
\frac{[1+m_A+d_{1,A}]}{[1+2m_A+d_{1,A}]}, & N_1=1, \\
\end{cases}\\
&&N_7:=\prod_{A\in S_R^+}
\prod_{i=0}^{N_3}\frac{[d_{2,A}+i]}{[d_{2,A}+m_A+i]}
\prod_{B\in S_M}
\frac{[d_{2,A}+m_A+c_{B}]}{[d_{2,A}+c_{B}]}
\end{eqnarray*}
where $c_{B}, B\in S_M$, is the sum of the number of arcs in $S_M$ 
right to $B$ or outside of $B$ (including $B$), and the number of 
dashed arcs right to $B$.
Suppose that $k$-the down arrow and the $l$-th ($k<l$) down arrow 
form a dashed arc $A$. 
Then $A$ is said to be size $(l-k+1)/2$.
Let $T'$ be the set of dashed arcs except the leftmost one.
We define
\begin{eqnarray*}
N_8:=
\begin{cases}
\prod_{A\in T}[d_{3,A}+m_A]^{-1}, & \text{for $N_1=1$}, \\
\prod_{A\in T'}[d_{3,A}+m_A]^{-1}, & \text{for $N_1=0$},
\end{cases}
\end{eqnarray*}
where $d_{3,A}:=(s-l+1)/2$.
Let $U$ be the union of $T$ and a down arrow with a star.
We define 
\begin{eqnarray*}
d_{4,A}&:=&
\begin{cases}
1+(N-s)/2, & \text{for a down arrow with a star}, \\
1+(N-k)/2, & \text{for a dashed arc},
\end{cases} \\
N_9&:=&
\frac{\prod_{i=1}^{N-N_2}(q^{i}+q^{-i})}
{\prod_{A\in U}(q^{d_{4,A}}+q^{-d_{4,A}})}
\end{eqnarray*}

We enumerate up arrows, arcs and dashed arcs and an unpaired down 
arrow from left. 
If there are arcs inside an arc or a dashed arc, we increase 
the integer one by one.
Let $N_A$ be the integer assigned to an arc or a dashed arc $A$,
and $N_\downarrow$ be the integer assigned to an unpaired down arrow.
We define
\begin{eqnarray*}
N_{10}:=
\begin{cases}
\prod_{A\in S\cup T}[N_A], & \text{for $N_1=0$}, \\
{[N_\downarrow]}\prod_{A\in S\cup T}[N_A] & \text{for $N_1=1$}.
\end{cases}
\end{eqnarray*}

\begin{example}	
Let $D$ be a diagram depicted as 
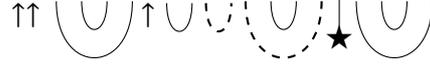
\begin{eqnarray*}
\uparrow\uparrow
\raisebox{-0.72\totalheight}{
\begin{tikzpicture}
\draw(0,0)..controls(0,-1)and(1,-1)..(1,0);
\draw(1/3,0)..controls(1/3,-0.5)and(2/3,-0.5)..(2/3,0);
\end{tikzpicture}}
\uparrow\! 
\raisebox{-0.5\totalheight}{
\begin{tikzpicture}
\draw(0,0)..controls(0,-0.5)and(1/3,-0.5)..(1/3,0);
\end{tikzpicture}}\,
\raisebox{-0.5\totalheight}{
\begin{tikzpicture}
\draw[dashed,thick](0,0)..controls(0,-0.5)and(1/3,-0.5)..(1/3,0);
\end{tikzpicture}}\,
\raisebox{-0.72\totalheight}{
\begin{tikzpicture}
\draw[dashed,thick](0,0)..controls(0,-1)and(1,-1)..(1,0);
\draw(1/3,0)..controls(1/3,-0.5)and(2/3,-0.5)..(2/3,0);
\end{tikzpicture}}\!\!\!
\raisebox{-0.62\height}{
\begin{tikzpicture}
\draw(0,0)--(0,-0.5)node{$\bigstar$};
\end{tikzpicture}}\!\!\!
\raisebox{-0.72\totalheight}{
\begin{tikzpicture}
\draw(0,0)..controls(0,-1)and(1,-1)..(1,0);
\draw(1/3,0)..controls(1/3,-0.5)and(2/3,-0.5)..(2/3,0);
\end{tikzpicture}}
\end{eqnarray*}
Then we have $N_2=9$, $N_5=[19]/[4], N_6=[13]/[14], N_7=[4][5]^{-1}[6]^{-1}, N_8=[3]^{-1}$, 
$N_{10}=[3]\cdot[4]\cdot[6]\cdot[7]\cdot[8]\cdot[9]\cdot[10]\cdot[11]$ and 
\begin{eqnarray*}
N_9=(q+q^{-1})(q^2+q^{-2})(q^4+q^{-4})\prod_{i=7}^{11}(q^i+q^{-i}).
\end{eqnarray*}  
\end{example}

\begin{theorem}
\label{thm-generic-psi}
In the above notation, we have 
\begin{eqnarray}
\label{generic-psi}
\Psi_D
=\prod_{A\in S}[m_A]^{-1}\cdot
N_5\cdot N_6 \cdot N_7 \cdot N_8\cdot N_9\cdot N_{10}.
\end{eqnarray}
\end{theorem}
Before proceeding to the proof of Theorem~\ref{thm-generic-psi}, 
we prove five lemmas used later.

Let $\mathbf{k}_i=(\underbrace{1,\ldots,1}_{N-i+1},\underbrace{-1,\ldots,-1}_{i-1})$
for $1\le i\le N+1$. 
Then, $\{\mathbf{k}_i| 1\le i\le n\}$ satisfies a unique sequence
\begin{eqnarray*}
\mathbf{k}_1=(1,\ldots,1)\rightarrow
\mathbf{k}_2\leftrightarrow
\mathbf{k}_3\rightarrow\cdots
\rightarrow \mathbf{k}_{N+1}=(-1,\ldots,-1).
\end{eqnarray*}
\begin{lemma}
\label{lemma-psi-1}
The components $\Psi_{i}:=\Psi_{\mathbf{k}_i}$, $1\le i\le N+1$, 
are 
\begin{eqnarray}
\label{Psi00}
\Psi_{2m+1}&=&\frac{\prod_{i=1}^{N-m}(q^{i}+q^{-i})}{\prod_{i=1}^{m}(q^{i}+q^{-i})}
\genfrac{[}{]}{0pt}{}{N-m}{m}, \\
\label{Psi01}
\Psi_{2m}&=&\frac{\prod_{i=1}^{N-m}(q^{i}+q^{-i})}{\prod_{i=1}^{m}(q^{i}+q^{-i})}
\frac{[N+1]}{[m]}
\genfrac{[}{]}{0pt}{}{N-m}{m-1}.
\end{eqnarray}
\end{lemma}
\begin{proof}
We use the notation in the proof of Theorem~\ref{thm-eigenvalues}.
The sequence is locally equivalent to 
\begin{eqnarray*}
\mathbf{k}_{2m-1}\rightarrow\mathbf{k}_{2m}
\leftrightarrow\mathbf{k}_{2m+1}.
\end{eqnarray*}
Since $\Psi$ is the eigenvector with the eigenvalue $[N+1]$, 
the components $\Psi_{i}:=\Psi_{\mathbf{k}_i}$, $1\le i\le N+1$, 
satisfy the following eigenvalue problem:
\begin{eqnarray*}
{[N-2m+2]}\Psi_{2m-1}+{[N-2m+1]}\Psi_{2m+1}&=&{[N+1]}\Psi_{2m}, \\
{[N-2m+1]}\Psi_{2m}&=&{[N+1]}\Psi_{2m+1}.
\end{eqnarray*}
The solution is 
\begin{eqnarray*}
\Psi_{2m}&=&\frac{[N+1][N-2m+2]}{[2m][2N-2m+2]}\Psi_{2m-1}, \\
\Psi_{2m+1}&=&\frac{[N-2m+1][N-2m+2]}{[2m][2N-2m+2]}\Psi_{2m-1}.
\end{eqnarray*}
Suppose $\Psi_{1}=\prod_{i=1}^{N}(q^{i}+q^{-i})$. 
The coefficients $\Psi_{k}$ satisfy Eqns.(\ref{Psi00}) and (\ref{Psi01}).
We have $\Psi_{N+1}=1$. 
This normalization is compatible with Definition~\ref{def-normalization}.
This completes the proof.
\end{proof}

Let $D$ be a diagram
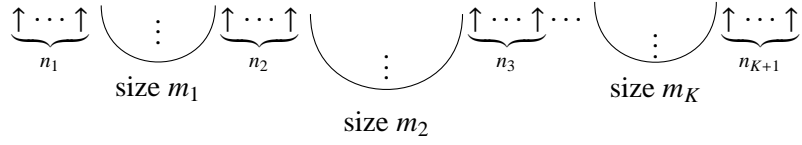
\begin{eqnarray*}
\underbrace{\uparrow\cdots\uparrow}_{n_1}
\raisebox{-0.72\totalheight}{
\begin{tikzpicture}[scale=0.5]
\draw(0,0)..controls(0,-2)and(3,-2)..(3,0);
\draw(1.5,-2.2)node{size $m_1$};
\draw(1.5,-0.5)node{$\vdots$};
\end{tikzpicture}}
\underbrace{\uparrow\cdots\uparrow}_{n_2}
\raisebox{-0.9\totalheight}{
\begin{tikzpicture}
\draw(0,0)..controls(0,-4/3)and(2,-4/3)..(2,0);
\draw(1,-1.45)node{size $m_2$};
\draw(1,-0.6)node{$\vdots$};
\end{tikzpicture}}
\underbrace{\uparrow\cdots\uparrow}_{n_3}
\cdots
\raisebox{-0.77\totalheight}{
\begin{tikzpicture}[scale=0.8]
\draw(0,0)..controls(0,-4/3)and(2,-4/3)..(2,0);
\draw(1,-1.45)node{size $m_K$};
\draw(1,-0.6)node{$\vdots$};
\end{tikzpicture}}
\underbrace{\uparrow\cdots\uparrow}_{n_{K+1}}
\end{eqnarray*}
where the region inside an arc of size $m_i, 1\le i\le K$, is filled 
with arcs.
Let $S_i$ be the set of all arcs inside the arc of size $m_{i}$.
\begin{lemma}\label{lemma-nodownarrow}
Let $D$ be a diagram depicted above.
The component $\Psi_D$ is 
\begin{eqnarray*}
\Psi_D
=
\prod_{i=1}^{d}(q^i+q^{-i})
\prod_{j=1}^{K}\frac{\left[\sum_{p=1}^{j}(n_p+m_p)\right]!}
{\left[n_j+\sum_{p=1}^{j-1}(n_p+m_p)\right]!}
\prod_{A\in S_j}[m_{A}]^{-1}
\end{eqnarray*}
where $d=\sum_{i=1}^{K+1}(m_i+n_i)$ with $m_{K+1}=0$.
\end{lemma}
\begin{proof}
We prove Theorem by induction.
When $D$ has no arc, that is, $n_1\neq0$, $n_i=0$, $2\le i\le K+1$,  
and $m_j=0$, $1\le j\le K$, we have $\Psi_{D}=\prod_{i=1}^{n_1}(q^{i}+q^{-i})$.
Theorem holds true since this is compatible with 
Definition~\ref{def-normalization} and Lemma~\ref{lemma-psi-1}.

We assume that Theorem is true for $D'$ where $D'$ has 
one less arcs than $D$.
From Eqn.(\ref{EP1}), the eigenvalue problem is equal to 
\begin{multline}
\label{EPDnostar}
\sum_{i=1}^{K}
\prod_{j=1}^{d+1}(q^{j}+q^{-j})\cdot
\left[1+\sum_{j=1}^{i}n_j\right]
\left(\prod_{j=1}^{i-1}L_j\right)L'_i
\left(\prod_{j=i+1}^{K}L''_j\right) 
+\left[1+\sum_{p=1}^{K+1}n_p\right]\Psi_{D}
=\left[1+\sum_{p=1}^{K+1}(n_p+2m_p)\right]\Psi_{D}
\end{multline}
where $m_{K+1}=0$ and 
\begin{eqnarray*}
L_j&=&\frac{\left[\sum_{p=1}^{j}(n_p+m_p)\right]!}
{\left[n_j+\sum_{p=1}^{j-1}(n_p+m_p)\right]!}
\prod_{A\in S_j}[m_A]^{-1}, \\
L'_j&=&\frac{\left[\sum_{p=1}^{j}(n_p+m_p)\right]!\cdot [m_j]}
{\left[1+n_j+\sum_{p=1}^{j-1}(n_p+m_p)\right]!}
\prod_{A\in S_j}[m_A]^{-1},\\
L''_j&=&\frac{\left[1+\sum_{p=1}^{j}(n_p+m_p)\right]!}
{\left[1+n_j+\sum_{p=1}^{j-1}(n_p+m_p)\right]!}
\prod_{A\in S_j}[m_A]^{-1}.
\end{eqnarray*}
The first term of the left hand side of Eqn.(\ref{EPDnostar}) is 
rewritten as 
\begin{eqnarray}
\label{EPDnostar2}
\prod_{j=1}^{d+1}(q^{j}+q^{-j})
\left(\prod_{j=1}^{K}L_j\right)\cdot
\sum_{i=1}^{K}
\left[1+\sum_{j=1}^{i}n_j\right][m_i]
\frac{\prod_{j:j>i}\left[1+\sum_{p=1}^{j}m_p+n_p\right]}
{\prod_{j:j\ge i}\left[1+n_j+\sum_{p=1}^{j-1}m_p+n_p\right]}.
\end{eqnarray}
Substituting Eqn.(\ref{EPDnostar2}) and Lemma~\ref{lemma-appendix1} into 
Eqn.(\ref{EPDnostar}), we obtain 
\begin{eqnarray*}
\Psi_D
=\prod_{i=1}^{d}(q^{i}+q^{-i})\cdot \prod_{i=1}^{K}L_i.
\end{eqnarray*}
This completes the proof.
\end{proof}

Let $D$ be a diagram depicted as
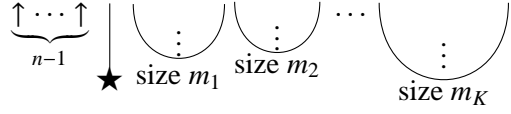
\begin{eqnarray*}
\underbrace{\uparrow\cdots\uparrow}_{n-1}
\!\!\!
\raisebox{-0.8\totalheight}{
\begin{tikzpicture}
\draw(0,0)--(0,-1)node{$\bigstar$};
\end{tikzpicture}}
\!\!\!
\raisebox{-0.8\totalheight}{
\begin{tikzpicture}[scale=0.6]
\draw(0,0)..controls(0,-1.6)and(2.0,-1.6)..(2,0);
\draw(1,-1.6)node{size $m_1$};
\draw(1,-0.7)node{$\vdots$};
\end{tikzpicture}}
\!\!\!
\raisebox{-0.75\totalheight}{
\begin{tikzpicture}[scale=0.56]
\draw(0,0)..controls(0,-1.6)and(2.0,-1.6)..(2,0);
\draw(1,-1.5)node{size $m_2$}(1,-0.68)node{$\vdots$};
\end{tikzpicture}}
\cdots
\raisebox{-0.82\totalheight}{
\begin{tikzpicture}[scale=0.85]
\draw(0,0)..controls(0,-1.6)and(2.0,-1.6)..(2,0);
\draw(1,-1.4)node{size $m_K$}(1,-0.7)node{$\vdots$};
\end{tikzpicture}}
\end{eqnarray*}
and $D'$ be the diagram obtained from $D$ by changing 
the rightmost up arrow of $D$ to a down arrow.
The region inside an arc of size $m_i, 1\le i\le K$, is filled 
with arcs.
Let $S$ be the set of all arcs.
We denote by $m_A$ the size of an arc $A$.
\begin{lemma}
\label{lemma-noarcleftstar}
In the above notation, we have 
\begin{eqnarray*}
\Psi_D
&=&
\frac{\prod_{i=1}^{n-1+M}(q^{i}+q^{-i})}{q^{M+1}+q^{-(M+1)}}
\frac{[n-1+M]!}
{[n-1]!\cdot\prod_{A\in S}[m_A]}
\frac{[1+n+2M]}{[1+M]},\\
\Psi_{D'}
&=&
\frac{[n-1]}{[1+n+2M]}\Psi_D,
\end{eqnarray*}
where $M=\sum_{i=1}^{K}m_i$.
\end{lemma}
\begin{proof}
There exists only one diagram $D''$, $D<_{\mathrm{lex}}D''$,  
such that the relation $D''\rightarrow D$ holds.
The diagram $D''$ is obtained from $D$ by changing the down arrow with 
a star to an up arrow. 
We have $\Psi_{D''}$ from Lemma~\ref{lemma-nodownarrow}.
The eigenvalue problem (\ref{EP1}) is equivalent to 
\begin{eqnarray*}
\prod_{i=1}^{n+M}(q^{i}+q^{-i})\cdot [n]
\frac{[n+M]!}{[n]!\prod_{A\in S}[m_A]}
+[n-1]\Psi_{D'}&=&[n+1+2M]\Psi_{D}, \\
{[n-1]}\Psi_{D}&=&[n+1+2M]\Psi_{D'}.
\end{eqnarray*}
Solving the above equations, we have a desired expression.
\end{proof}

Let $D$ be a diagram depicted as 
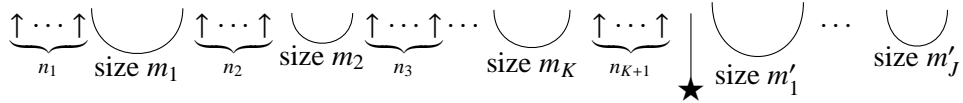
\begin{eqnarray*}
\underbrace{\uparrow\cdots\uparrow}_{n_1}
\!\!\!
\raisebox{-0.7\totalheight}{
\begin{tikzpicture}[scale=0.4]
\draw(0,0)..controls(0,-2)and(3,-2)..(3,0);
\draw(1.5,-2)node{size $m_1$};
\end{tikzpicture}}
\underbrace{\uparrow\cdots\uparrow}_{n_2}
\!\!\!
\raisebox{-0.7\totalheight}{
\begin{tikzpicture}[scale=0.4]
\draw(0,0)..controls(0,-4/3)and(2,-4/3)..(2,0);
\draw(1,-1.45)node{size $m_2$};
\end{tikzpicture}}
\!\!\!
\underbrace{\uparrow\cdots\uparrow}_{n_3}
\cdots
\!\!\!
\raisebox{-0.7\totalheight}{
\begin{tikzpicture}[scale=0.5]
\draw(0,0)..controls(0,-4/3)and(2,-4/3)..(2,0);
\draw(1,-1.45)node{size $m_K$};
\end{tikzpicture}}
\underbrace{\uparrow\cdots\uparrow}_{n_{K+1}}
\!\!\!
\raisebox{-0.8\totalheight}{
\begin{tikzpicture}
\draw(0,0)--(0,-1)node{$\bigstar$};
\end{tikzpicture}
}
\!\!\!\!\!
\raisebox{-0.7\totalheight}{
\begin{tikzpicture}[scale=0.6]
\draw(0,0)..controls(0,-1.6)and(2.0,-1.6)..(2,0);
\draw(1,-1.7)node{size $m'_1$};
\end{tikzpicture}}
\,
\cdots
\raisebox{-0.7\totalheight}{
\begin{tikzpicture}[scale=0.4]
\draw(0,0)..controls(0,-1.6)and(2.0,-1.6)..(2,0);
\draw(1,-1.7)node{size $m'_{J}$};
\end{tikzpicture}}
\end{eqnarray*}
and $D'$ be the diagram obtained from $D$ by changing 
the rightmost up arrow of $D$ to a down arrow.
Let $S$ be the set of all arcs inside of 
an arc of size $m_i, 1\le i\le K$ or $m'_j, 1\le j\le J$.  
We consider two cases: (1) $n_{K+1}\neq0$ (Lemma~\ref{lemma-link21})
and (2) $n_{I+1}\neq0$ and $n_{i}=0$, $I+2\le i\le K+1$ (Lemma~\ref{lemma-link22}).
\begin{lemma}
\label{lemma-link21}
Let $n_{K+1}\neq0$. 
Set $N_i=\sum_{j=1}^{i}n_i$, $M_i=\sum_{j=1}^{i}m_j$, 
$N=N_{K+1}, M=M_K$ and $M'=\sum_{i=1}^Jm'_i$.
We have
\begin{eqnarray}
\label{onelink21}
\Psi_{D}
&=&
\frac{\prod_{i=1}^{N+M+M'}(q^{i}+q^{-i})}{(q^{M'+1}+q^{-(M'+1)})}
\frac{[2+N+M+2M']}
{\prod_{A\in S}[m_A]\cdot[M'+1]}
\prod_{i=1}^{K}
\frac{[N_i+M_i]!}{[N_i+M_{i-1}]!}
\frac{[N+M+M']!}{[N+M]!},
\\
\Psi_{D'}
&=&\frac{[N+M]}{[2+N+M+2M']}\Psi_D.
\end{eqnarray}

\end{lemma}
\begin{proof}
We prove Lemma by induction. 
The case where $m_i=0$ for $1\le i\le K$ is reduced to Lemma~\ref{lemma-noarcleftstar}.
We assume that Lemma holds true for a diagram which has one less arcs than $D$.
Let $D_1$ be a diagram obtained from $D$ by changing the outer arc of size $m_i$ 
to two up arrows.
Then, we have $D_1\rightarrow D$. 
If $D_2$ is obtained from $D$ by changing the down arrow with a star to an up 
arrow, we have $D_2\rightarrow D$.
Let $D'$ be a diagram obtained from $D$ by changing the rightmost up arrow of 
$D$ to an unpaired down arrow. 
We have $D\leftrightarrow D'$. 
Let $D_3$ be a diagram obtained from $D'$ by changing the outer arc of 
size $m_i$ to two up arrows.
We have $D_3\rightarrow D'$.

Let $A$ be the right hand side of Eqn.(\ref{onelink21}).
We have an explicit expression for $\Psi_{D_1}$ from the induction 
assumption, and $\Psi_{D_2}$ from Lemma~\ref{lemma-noarcleftstar}.
Therefore, the eigenvalue problem (\ref{EP1}) for the $D$- and 
$D'$-components is reduced to 
\begin{multline*}
\sum_{i=1}^{K}
A(q^{d}+q^{-d})[1+N_i][m_i]
\frac{\prod_{j\ge i+1}[1+N_j+M_j]}
{\prod_{j\ge i}[1+N_j+M_{j-1}]}
\frac{[N+M+M'+1][N+M+2M'+3]}{[N+M+1][N+M+2M'+2]}  \\
+A(q^{d}+q^{-d})(q^{M'+1}+q^{-(M'+1)})
\frac{[N+M+M'+1][N+1][M'+1]}{[N+M+1][N+M+2M'+2]}
+[N]\Psi_{D'} \\ 
=[N+2M+2M'+2]\Psi_{D}, 
\end{multline*}
\begin{multline*}
\sum_{i=1}^{K}
A(q^{d}+q^{-d})[1+N_i][m_i]
\frac{\prod_{j\ge i+1}[1+N_j+M_j]}
{\prod_{j\ge i}[1+N_j+M_{j-1}]}
\frac{[N+M+M'+1]}{[2+N+M+2M']} \\
+[N]\Psi_{D}=[N+2M+2M'+2]\Psi_{D'}.
\end{multline*}
where $d=N+M+M'+1$.
Substituting Lemma~\ref{lemma-appendix1} into the above equations, 
we obtain
\begin{eqnarray*}
\Psi_D=A, \qquad\Psi_{D'}=\frac{[N+M]}{[2+N+M+2M']}A.
\end{eqnarray*}
This completes the proof.
\end{proof}
\begin{lemma}
\label{lemma-link22}
Let $n_{I+1}\neq0$ and $n_{i}=0$ for $I+2\le i\le K+1$. 
Set $N=\sum_{i=1}^{I+1}n_i, M_1=\sum_{i=1}^{I}m_i, M_2=\sum_{i=I+1}^{K}m_i$ 
and $M'=\sum_{i=1}^{J}m'_i$.
We have 
\begin{eqnarray}
\label{onelink22}
&&\Psi_D
=
\frac{\prod_{i=1}^{d}(q^{i}+q^{-i})}{q^{M'+1}+q^{-(M'+1)}}
\prod_{A\in S}[m_A]^{-1}
\prod_{i=1}^{I}
\frac{[\sum_{j=1}^{i}n_j+m_j]!}{[\sum_{j=1}^{i}n_j+\sum_{j=1}^{i-1}m_j]!}
\frac{[d]!}{[M_1+N]!}
\frac{[L_1]}{[M'+1]}
\prod_{i=I+1}^{K}
\frac{[g_i-m_i]}{[g_i]}, \\
&&\Psi_{D'}=\frac{[N+M_1]}{[L_1]}\Psi_D,
\end{eqnarray}
where $d=N+M_1+M_2+M'$, $L_1=N+M_1+2M_2+2M'+2$ 
and $g_i=2+2M'+2\sum_{j=i}^{K}m_j$.
\end{lemma}
\begin{proof}
We prove Lemma by induction.
We consider the case where $m_i=0$ for $1\le i\le K-1$ and $m_K\neq0$. 
If $D_1$ has one less arcs than $D$ ($D_1$ does not have the outer arc
of size $m_K$), we have $D_1\rightarrow D$.
If $D_2$ does not have a down arrow with a star but all arcs are in 
the same position as $D$, then $D_2\rightarrow D$.
We have a formula for $\Psi_{D_1}$ and $\Psi_{D_2}$ from 
Lemma~\ref{lemma-nodownarrow} and Lemma~\ref{lemma-link21}.
On the other hand, a diagram $D_3$ such that $D_3\rightarrow D'$
is only $D_3=D$.
Let $A$ be the right hand side of Eqn.(\ref{onelink22}).
The eigenvalue problem is reduced to 
\begin{multline*}
A(q^{d+1}+q^{-(d+1)})[m_K]\frac{[N+m_K+M'+1]}{[N+m_K+1]}\frac{[2+2m_K+2M']}{[2+m_K+2M']}
\frac{[N+m_K+2M'+3]}{[N+2m_K+2M'+2]} \\
+A(q^{d+1}+q^{-(d+1)})(q^{M'+1}+q^{-(M'+1)})
\frac{[N+1][N+m_K+M'+1][M'+1][2+2m_K+2M']}{[N+m_K+1][N+2m_K+2M'+2][2+m_K+2M']} \\
+[N]\Psi_{D'}=[N+2+2m_K+2M']\Psi_{D}, 
\end{multline*}
\begin{eqnarray*}
[N]\Psi_{D}=[N+2+2m_K+2M']\Psi_{D'}
\end{eqnarray*}
The solution is 
\begin{eqnarray*}
\Psi_D=A,\qquad \Psi_{D'}=\frac{[N]}{[N+2+2m_K+2M']}A
\end{eqnarray*}
which is a desired expression.

We assume that Lemma holds true for a diagram which has one less 
arcs than $D$.
If $D_1$ has one less arcs than $D$, $D_1\rightarrow D$.
If $D_2$ does not have a down arrow with a star and all arcs are 
in the same position as $D$, $D_2\rightarrow D$.
We also have $D'\rightarrow D$.
From the assumption, Lemma~\ref{lemma-nodownarrow} and 
Lemma~\ref{lemma-link21}, we have a formula for the components 
$\Psi_{D_1}$ and $\Psi_{D_2}$.
Similarly, if $D_3$ has one less arcs than $D'$, $D_3\rightarrow D'$.
We have $D\rightarrow D'$.
Let $A$ be the right hand side of Eqn.(\ref{onelink22}).
The eigenvalue problem (\ref{EP1}) is reduced to 
\begin{multline}
\label{onelink24}
\sum_{i=1}^{I}A(q^{d+1}+q^{-(d+1)})\left[1+\sum_{j=1}^{i}n_j\right]
\frac{\prod_{j\ge i+1}[1+\sum_{k=1}^{j}n_k+m_k]}
{\prod_{j\ge i}[1+\sum_{k=1}^{j}n_k+\sum_{k=1}^{j-1}m_k]}
\frac{[m_i][d+1]}{[N+M_1+1]}\frac{[L_1+1]}{[L_1]} \\
+\sum_{i=I+1}^{K}(q^{d+1}+q^{-(d+1)})A
\frac{\prod_{j\ge i+1}[1+N+M_1+\sum_{k=I+1}^{j}m_k]}
{\prod_{j\ge i}[1+N+M_1+\sum_{k=I+1}^{j-1}m_k]}
\frac{[m_i][N+1][d+1]}{[1+N+M_1+M_2]}
\frac{[h_i]}{[L_1]}
\prod_{j=I+1}^{i}\frac{[g_i]}{[g_i-m_i]} \\
+A(q^{d+1}+q^{-(d+1)})(q^{M'+1}+q^{-(M'+1)})[N+1]\frac{[d+1]}{[N+M_1+M_2+1]}
\frac{[M'+1]}{[L_1]}
\prod_{j=I+1}^{K}\frac{[g_j]}{[g_j-m_j]}  \\
+[N]\Psi_{D'}=[L_1+M_1]\Psi_{D}
\end{multline}
where $h_i=3+N+\sum_{j=I+1}^{i}m_j+2\sum_{j=i+1}^{K}m_p+2M'$ and 
\begin{multline}
\label{onelink25}
\sum_{i=1}^{I}A(q^{d+1}+q^{-(d+1)})\left[1+\sum_{j=1}^{i-1}n_j\right]
\frac{\prod_{j=i+1}^{I}[1+\sum_{k=1}^{j}n_k+m_k]}
{\prod_{j=i}^{I}[1+\sum_{k=1}^{j}n_k+\sum_{k=1}^{j-1}m_k]}
\frac{[d+1]}{[L_1]} 
+[N]\Psi_{D}=[L_1+M_1]\Psi_{D'}.
\end{multline}

We apply Lemma~\ref{lemma-appendix1} to the first term of 
Eqn.(\ref{onelink24}), Lemma~\ref{lemma-appendix2} to the 
second and third terms with $x=N+M_1$ and $z=M'$, and 
Lemma~\ref{lemma-appendix1} to Eqn.(\ref{onelink25}).
We obtain
\begin{eqnarray*}
(q^{d+1}+q^{-(d+1)})\frac{[d+1][L_1-N]}{[L_1]}A+[N]\Psi_{D'}&=&[L_1+M_1]\Psi_{D}, \\
(q^{d+1}+q^{-(d+1)})A[M_1]\frac{[d+1]}{[L_1]}
+[N]\Psi_{D}&=&[L_1+M_1]\Psi_{D'}.
\end{eqnarray*}
The solution is 
\begin{eqnarray*}
\Psi_D=A,\qquad \Psi_{D'}=\frac{[N+M_1]}{[L_1]}A.
\end{eqnarray*}
This completes the proof.
\end{proof}

\begin{proof}[Proof of Theorem~\ref{thm-generic-psi}]
We prove Theorem by induction on the numbers of arcs and dashed arcs.
From Lemma~\ref{lemma-psi-1}, Lemma~\ref{lemma-link21} and 
Lemma~\ref{lemma-link22}, Theorem holds true when the number of 
dashed arcs or arcs is zero.

Let $D_0$ be a diagram which has a down arrow with a star and $N_1=0$. 
The diagram $D_0$ starts with $n_1$ up arrows, followed by an outer 
arc of size $m_1$, followed by $n_2$ up arrows, followed by an outer 
arc of size $m_2$, $\cdots$, followed by $n_{I+1}$ up arrows, followed 
by a dashed arc, $\cdots$, followed by a down arrow with a star, 
followed by an arc of size $m'_1$, $\cdots$ and ends with an arc of 
size $m'_J$.
Let $D'_0$ be a diagram obtained from $D_0$ by changing the rightmost
up arrow of $D_0$ to an unpaired down arrow. 
We assume that Eqn.(\ref{generic-psi}) holds true for a diagram which 
has one less dashed arcs or one less arcs than $D_0$.
Let $N'=\sum_{i=1}^{I+1}n_i$, $M=\sum_{i=1}^{I}m_i$, 
$L_1=N'+M+2|S_M|+2|T|+2|S_R|+2$, $L_2=N'+M+|S_M|+|T|+|S_R|+1$ and 
$A$ be the right hand side of Eqn.(\ref{generic-psi}) for $D_0$

We consider the two cases: 1) $n_{I+1}\neq0$ and 
2) $n_{H+1}\neq0, n_{i}=0, H+2\le i\le I+1$. 

\paragraph{\bf Case 1}
If $D_1$ has one less arcs than $D_0$, then $D_1\rightarrow D_0$.
If $D_2$ has one less dashed arcs than $D_0$, then $D_2\rightarrow D_0$.
We have $D'_0\rightarrow D_0$.
Similarly, if $D_3$ has one less arcs than $D'_0$, then 
$D_3\rightarrow D'_0$. 
We have $D_{0}\rightarrow D'_0$.
The eigenvalue problem is equivalent to 
\begin{multline}
\label{generic-psi2}
(q^{L_2}+q^{-L_2})A\sum_{i=1}^{I}\left[1+\sum_{j=1}^{i}n_j\right][m_i]
\frac{\prod_{j=i+1}^{I}[1+\sum_{k=1}^{j}(n_k+m_k)]}
{\prod_{j=i}^{I}[1+\sum_{k=1}^{j}n_j+\sum_{k=1}^{j-1}m_k]}
\frac{[L_2][1+L_1]}{[1+N'+M][L_1]} \\
+(q^{L_2}+q^{-L_2})(q^{d}+q^{-d})[N'+1]A\frac{[L_2][1+|S_R|+|S_M|+|T|]}{[L_1][N'+M+1]}
+[N']\Psi_{D'_0}=[L_1+M]\Psi_{D_0},
\end{multline}
where  $d=|S_R|+|S_M|+|T|+1$ and
\begin{eqnarray*}
(q^{L_2}+q^{-L_{2}})A[M]\frac{[L_2]}{[L_1]}+[N']\Psi_{D_0}=[L_1+M]\Psi_{D'_0},
\end{eqnarray*}
where we have used the fact that the weight of the leftmost 
dashed arc in $D'_0$ is $(1+|S_R|+|T|)^{-1}$.
Applying Lemma~\ref{lemma-appendix1} to Eqn.(\ref{generic-psi2}),
we obtain 
\begin{eqnarray*}
\Psi_{D_0}=A,\qquad \Psi_{D'_0}=\frac{[N'+M]}{[L_1]}A.
\end{eqnarray*}
This completes the proof of Case 1.

\paragraph{\bf Case 2}
Set $M_1=\sum_{i=1}^{H}m_i$.
By a similar argument to Case 1, the eigenvalue problem for $\Psi_D$ 
is equal to 
\begin{multline}
\label{eqn-genericPsi-21}
(q^{L_2}+q^{-L_2})A\sum_{i=1}^{H}\left[1+\sum_{j=1}^{i}n_j\right][m_i]
\frac{\prod_{j\ge i+1}[1+\sum_{k=1}^{j}n_k+m_k]}
{\prod_{j\ge i}[1+\sum_{k=1}^{j}n_k+\sum_{k=1}^{j-1}m_k]}
\frac{[L_2][1+L_1+M-M_1]}{[1+N'+M_1][L_1-M_1+M]}\\
+
(q^{L_2}+q^{-L_2})A\sum_{i=H+1}^{K}
\frac{\prod_{j\ge i+1}[1+M_1+N'+\sum_{k=H+1}^{j}m_k]}
{\prod_{j\ge i}[1+M_1+N'+\sum_{k=H+1}^{j-1}m_k]}
\frac{[N'+1][L_2][L_{3,i}]}{[1+M+N'][L_1-M_1+M]}
\prod_{j=H+1}^{i}\frac{[g_j]}{[g_j-m_j]} \\
+(q^{L_2}+q^{-L_2})(q^d+q^{-d})A[N'+1]\frac{[L_2][1+|S_R|+|S_M|+|T|]}{[N'+M+1][L_1-M_1+M]}
\prod_{i=H+1}^{K}\frac{[g_i]}{[g_i-m_i]} \\
+[N']\Psi_{D'}=[L_1+M]\Psi_D,
\end{multline}
where $L_{3,i}=1+L_1+\sum_{j=i+1}^{K}m_j$,   
$g_i=2+2\sum_{j=i}^{K}m_j+2|S_R|+2|S_M|+2|T|$ and $d=|S_R|+|S_M|+|T|+1$.
Applying Lemma~\ref{lemma-appendix1} to the first term
and Lemma~\ref{lemma-appendix2} to the second and third 
terms of Eqn.(\ref{eqn-genericPsi-21}), we obtain
\begin{eqnarray}
\label{generic-psi3}
\frac{A[2L_2][2d+2M-M_1]}{[L_1+M-M_1]}
+[N']\Psi_{D'}=[L_1+M]\Psi_{D}. 
\end{eqnarray}
The eigenvalue problem for $\Psi_{D'}$ is 
\begin{eqnarray*}
A\sum_{i=1}^{H}[m_i]
\frac{\prod_{j\ge i+1}[1+\sum_{k=1}^{j}n_k+m_k]}
{\prod_{j\ge i}[1+\sum_{k=1}^{j}n_k+\sum_{k=1}^{j-1}m_k]}
\frac{[2L_2]}{[L_1+M-M_1]} 
+[N']\Psi_{D}=[L_1+M]\Psi_{D'}.
\end{eqnarray*}
Inserting Lemma~\ref{lemma-appendix1}, we obtain
\begin{eqnarray}
\label{generic-psi4}
A\frac{[M_1][2L_2]}{[L_1+M-M_1]}
+[N']\Psi_{D}=[L_1+M]\Psi_{D'}.
\end{eqnarray}
The solution of Eqns.(\ref{generic-psi3}) and 
(\ref{generic-psi4}) is
\begin{eqnarray*}
\Psi_D=A,\qquad \Psi_{D'}=\frac{[N'+M_1]}{[L_1+M-M_1]}A.
\end{eqnarray*}
This completes the proof.
\end{proof}

We have two propositions regarding the eigenfunction $\Psi$. 
\begin{prop}
\label{prop-pint-psi}
All components of $\Psi$ belong to $\mathbb{N}[q,q^{-1}]$ and 
invariant under $q\rightarrow q^{-1}$.
\end{prop}
\begin{proof}
Let $\mathbf{k}\in I_{\mathbf{m}}$. 
From Theorem~\ref{thm-generic-psi}, the eigenvector $\Psi_{\mathbf{k}}$ 
contains only quantum numbers and $(q^{i}+q^{-i})$ for some $i$.
Therefore, $\Psi_{D}$ is obviously invariant under $q\rightarrow q^{-1}$.
We abbreviate $v^{k_1}\otimes\cdots v^{k_{n}}$ as $v^{\mathbf{k}}$. 
Recall that $\Psi^{0}_{\kappa}$, $\kappa=\{\pm1\}^{N}$, is the 
eigenvector of $Y$ with the eigenvalue $[N+1]$. 
Since the action of $Y$ commutes with the one of the projection 
$\pi:=\pi_{m_1}\otimes\cdots\otimes\pi_{m_n}$, the eigenvector 
$\Psi^{0}_{\mathbf{k}}$ on the standard basis is written as 
$\Psi^{0}_{\mathbf{k}}v^{\mathbf{k}}
=\sum_{\kappa}\pi(\Psi^{0}_{\kappa}v^{\kappa})$ 
where the sum is all over $\kappa$'s satisfying 
$v^{\mathbf{k}}\propto\pi(v^{\kappa})$.
The transition matrix from the standard basis to the Kazhdan--Lusztig 
basis is nothing but $R:=R_{\mathbf{k,l}}$. 
The two vectors $\Psi$ and $\Psi^{0}$ are a unique eigenvector of $Y$ 
with the multiplicity one. 
Therefore, we have $\Psi=R\Psi^{0}$. 
From Corollary~\ref{cor-RP}, we have 
$R_{\mathbf{k,l}}\in q^{-1}\mathbb{N}[q^{-1}]$ for 
$\mathbf{k}\neq\mathbf{l}$ and $R_{\mathbf{k,k}}=1$. 
Since the vector $\Psi^{0}$ is in $\mathbb{N}[q]$, 
we have $\Psi_{\mathbf{k}}\in\mathbb{N}[q,q^{-1}]$.
\end{proof}

\begin{remark}
In the proof of Proposition~\ref{prop-pint-psi}, we show that 
$\Psi=R\Psi^{0}$. 
This gives highly non-trivial relations among the 
Kazhdan--Lusztig polynomials.
\end{remark}

Let $\mathbf{k}\in I_{\mathbf{m}}$. 
Recall the diagram for a standard base 
$v^{k_1}\otimes\ldots\otimes v^{k_{N}}$.
Let $J_{\mathbf{k}}$ be the set of positions (from the right end) 
of up arrows of the diagram.
We define 
\begin{eqnarray*}
d_{\mathbf{k}}:=\sum_{i\in J_{\mathbf{k}}}i.
\end{eqnarray*}
For example, when $\mathbf{k}=(0,2)$ with $\mathbf{m}=(2,2)$, 
$J_{\mathbf{k}}=\{1,2,4\}$ and $d_{\mathbf{k}}=1+2+4=7$.
\begin{prop}
The component $\Psi_{\mathbf{k}}$
has the leading term $q^{d_{\mathbf{k}}}$ 
with the leading coefficient one.
\end{prop}
\begin{proof}
Notice that if $\mathbf{k}<_{\mathrm{lex}}\mathbf{k'}$, we have 
$d_{\mathbf{k'}}<d_{\mathbf{k}}$.
The vector $\Psi^{0}_{\mathbf{k}}$ has a leading term $q^{d_{\mathbf{k}}}$ 
with the coefficient one. 
The matrix representation of $R$ is an upper triangular matrix 
whose diagonal entries are one and non-zero entries are in 
$q^{-1}\mathbb{N}[q^{-1}]$. 
In the proof of Proposition~\ref{prop-pint-psi}, we show that 
$\Psi=R\Psi^{0}$. 
Therefore, the leading term of $\Psi_{\mathbf{k}}$ is $q^{d_{\mathbf{k}}}$
and the leading coefficient is one. 
\end{proof}

\subsection{Sum rule}
Below, we set $m_i=m$, $1\le i\le n$.
We denote by $s_{m,L}:=\sum_{\mathbf{k}}\Psi_{\mathbf{k}}$ 
the sum of components of $\Psi$.
We are interested in the case of $q=1$ since we expect that 
the sum of the eigenvector $\Psi$ is related to the total number 
of some combinatorial objects.
In Table~\ref{Table-sum}, we list up first few values of $s_{m,L}$
at $q=1$.

The Pell numbers $P_{n}$ are defined by the recurrence relation 
\begin{eqnarray*}
P_{n}=2P_{n-1}+P_{n-2},
\end{eqnarray*}
with $P_{0}=1$ and $P_{1}=1$. 
Let $c_{n}$ be the sequence of integers A094723 in \cite{OEIS}. 
The sequence $c_n$ is given by 
\begin{eqnarray*}
c_{n}=P_{n+2}-2^{n}.
\end{eqnarray*}
\begin{theorem}[Sum Rule I]
The sum $s_{m,1}$ at $q=1$ satisfy 
\begin{eqnarray*}
s_{m,1}=c_{m}.
\end{eqnarray*}
\end{theorem}
\begin{proof}
Since we have $L=1$, diagrams for the dual canonical bases have no arcs.
The number of diagrams is $m+1$. 
The graph $\Gamma$ is equal to 
\begin{eqnarray*}
\raisebox{-0.4\totalheight}{
\begin{tikzpicture}
\draw[thick,->](0,-0.2)..controls(-0.7,0.3)and(-0.7,-0.8)..(0,-0.3);
\end{tikzpicture}}
m\rightarrow m-1 \leftrightarrow
m-2\rightarrow\cdots
\rightarrow -m.
\end{eqnarray*}
The graph $\Gamma$ is isomorphic to the sequence considered in 
Lemma~\ref{lemma-psi-1}.
Therefore, the component $\Psi_{k}, k\in I_{m}$, is the same as in
Lemma~\ref{lemma-psi-1}.
The sum $s_{m,1}$ is written as 
\begin{eqnarray*}
s_{m,1}&=&\sum_{n=0}^{\lfloor m/2\rfloor}2^{m-2n}\genfrac{(}{)}{0pt}{}{m-n}{n}
+\sum_{n=1}^{\lfloor (m+1)/2\rfloor}2^{m-2n}\frac{m+1}{n}\genfrac{(}{)}{0pt}{}{m-n}{n-1} \\
&=&-2^{m}+\sum_{n=0}^{\infty}2^{m-2n+1}\genfrac{(}{)}{0pt}{}{m-n+1}{n}.
\end{eqnarray*}
Let 
\begin{eqnarray*}
A_m=\sum_{n=0}^{\infty}2^{m-2n+1}\genfrac{(}{)}{0pt}{}{m-n+1}{n}.
\end{eqnarray*}
Then $A_m$ satisfies the recurrence relation $A_{m+1}=2A_{m}+2A_{m-1}$
with $A_0=2$ and $A_1=5$. 
This implies $A_m=P_{m+2}$.
Therefore, $s_{m,1}=P_{m+2}-2^{m}$. 
This completes the proof.
\end{proof}

\begin{table}[ht]
\caption{First few values of the sum $s_{m,L}$ at $q=1$}
\label{Table-sum}
\centering
\scalebox{0.95}{
\begin{tabular}{cc|rrrrrr}
 & & \multicolumn{5}{c}{Length $L$} \\ 
& & 1 & 2 & 3 & 4 & 5 & 6 \\ \hline
\multicolumn{1}{c}{\multirow{9}{*}{$m$}} & 1 & 3 & 10 & 
38 & 156 & 692 & 3256 \\ \cline{2-8}
& 2 & 8 & 92 & 1408 & 26576 & 594432  & \\ \cline{2-8}
& 3 & 21 & 832 & 52736 & 4700592 & 549144752 &  \\ \cline{2-8}
& 4 & 54 & 7276 & 1924040 & 817051024 & 505001670752 & $\cdots$ \\ \cline{2-8}
& 5 & 137 & 62756 & 69395300 & 141326485016 & 4184410893902752 &  \\ \cline{2-8} 
& 6 & 344 & 534416 & 2479324096 & 24339640457600 & 430183061610221568 &  \\ \cline{2-8}
& 7 & 857 & 4514352 & 88070572208 &4184410893902752 & 398477790183643039008 &  \\ \cline{2-8} 
& 8 & 2122 & & \scalebox{0.8}{$\vdots$} & & &  \\ \cline{1-8}
\end{tabular}}
\end{table}

Let $c_n$ be the sequence of integers A00902 in~\cite{OEIS} (see also~\cite{Rob76}). 
The sequence $c_n$ satisfies the following recurrence relation:
\begin{eqnarray*}
c_n=2c_{n-1}+(2n-2)c_{n-2},
\end{eqnarray*}
with the initial condition $c_1=1$ and $c_2=3$.
\begin{conj}[Sum Rule II]
The sum $s_{1,L}$ at $q=1$ satisfy 
\begin{eqnarray*}
s_{1,L}=c_{L+1}
\end{eqnarray*}
for all $L\ge1$.
\end{conj}
We have checked the conjecture up to $L=24$.
The sum is conjectured to be equal to the total number of arrangements
of bishops in $2n\times 2n$ with a certain symmetry.
It would be interesting if we can find a combinatorial meaning of each 
component $\Psi_{\mathbf{k}}$.

\appendix
\section{}
\label{Sec:app}
\begin{lemma}
\label{lemma-appendix1}
\begin{eqnarray}
\label{appendix-0}
\sum_{i=1}^{K}
\left[1+\sum_{j=1}^{i}n_j\right][m_i]
\frac{\prod_{j\ge i+1}[1+\sum_{k=1}^{j}(n_k+m_k)]}
{\prod_{j\ge i}[1+\sum_{k=1}^{j}n_k+\sum_{k=1}^{j-1}m_k]}
=\left[\sum_{i=1}^{K}m_i\right]
\end{eqnarray}
\end{lemma}
\begin{proof}
We prove Lemma by induction on $K$.
Let $f^K$ be the left hand side of Eqn.(\ref{appendix-0}).
It is obvious that $f^1=[m_1]$.
We assume Eqn.(\ref{appendix-0}) is true for $K-1$.
Then, we have 
\begin{eqnarray*}
f^K&=&f^{K-1}\frac{[1+\sum_{j=1}^{K}(n_j+m_j)]}
{[1+\sum_{j=1}^{K}n_j+\sum_{j=1}^{K-1}m_j]}
+
\frac{[m_K][1+\sum_{j=1}^{K}n_j]}{[1+\sum_{j=1}^{K}n_j+\sum_{j=1}^{K-1}m_j]} \\
&=&\left[\sum_{i=1}^{K}m_i\right].
\end{eqnarray*}
This completes the proof.
\end{proof}

\begin{lemma}
\label{lemma-appendix2}
\begin{eqnarray}
\label{appendix-1}
\sum_{i=1}^{K}
I_i\cdot J_i
+\frac{[2z+2]}{[x+1+\sum_{i=1}^{K}m_i]}I_K
=[1+x]^{-1}\left[2+2z+2\sum_{i=1}^{K}m_i\right],
\end{eqnarray}
where 
\begin{eqnarray*}
I_i&:=&
\prod_{j=1}^{i}\frac{[2+2m_j+2z+2\sum_{k\ge j+1}^{K}m_k]}
{[2+m_j+2z+2\sum_{k=j+1}^{K}m_{k}]}, \\
J_i&:=&\frac{[m_i][x+3+\sum_{j=1}^{i}m_j+2\sum_{j=i+1}^{K}m_j+2z]}
{[1+x+\sum_{j=1}^{i-1}m_j][1+x+\sum_{j=1}^{i}m_j]}.
\end{eqnarray*}
\end{lemma}
\begin{proof}
We prove Lemma by induction on $K$. 
Let $f^{K}(z)$ be the left hand side of Eqn.(\ref{appendix-1}).
We have $f^{1}(z)=[1+x]^{-1}[2+2z+2m_1]$.
We assume that Lemma holds true for $f^{K-1}(z)$.
\begin{eqnarray*}
f^{K}(z)&=&f^{K-1}(z+2m_K)
-\frac{[2z+2m_K+2]}{[x+1+\sum_{i=1}^{K-1}m_i]}
I_{K-1}
+J_KI_K
+\frac{[2z+2]}{[x+1+\sum_{i=1}^{K}m_i]}I_K \\
&=&
[1+x]^{-1}\left[2+2z+2\sum_{i=1}^{K}m_i\right].
\end{eqnarray*}
This completes the proof.
\end{proof}

\bibliographystyle{amsplainhyper} 
\bibliography{biblio}

\end{document}